\documentclass[journal]{IEEEtran}
\usepackage{amsmath}    
\usepackage{graphics,epsfig,subfigure}    
\usepackage{verbatim}   
\usepackage{color}      

\usepackage{subfigure}  
\usepackage{hyperref}   
\usepackage{amssymb}
\usepackage{epstopdf}
\usepackage{graphicx}
\usepackage{wasysym}
\usepackage{wrapfig}
\usepackage{sidecap}
\usepackage{float}
\usepackage{subfigure}
\usepackage{enumerate}
\usepackage{graphicx}
\graphicspath{{./figures/}}
\usepackage{algorithmic}
\usepackage{algorithm}
\usepackage{multirow}
\usepackage{amssymb}

\allowdisplaybreaks

\newtheorem{theorem}{Theorem}
\newtheorem{lemma}{Lemma}
\newtheorem{problem}{Problem}
\newtheorem{definition}{Definition}

\begin{document}
%
\title{Autonomous Vehicle Public Transportation System: Scheduling and Admission Control}

\author{Albert Y.S. Lam,
        Yiu-Wing Leung, and
        Xiaowen Chu
\thanks{A preliminary version of this paper was presented in \cite{confversion}.}
\thanks{A.Y.S. Lam is with the Department of Electrical and Electronic Engineering, The University of Hong Kong, Pokfulam,
Hong Kong (e-mail: ayslam@eee.hku.hk).}
\thanks{Y.-W Leung and X. Chu are with the Department
of Computer Science, Hong Kong Baptist University, Kowloon Tong,
Hong Kong (e-mail: \{ywleung, chxw\}@comp.hkbu.edu.hk).}
}

%


\maketitle

\begin{abstract}
Technology of autonomous vehicles (AVs) is getting mature and many AVs will appear on the roads in the near future. AVs become connected with the support of various vehicular communication technologies and they possess high degree of control to respond to instantaneous situations cooperatively with high efficiency and flexibility. In this paper, we propose a new public  transportation system based on AVs. It manages a fleet of AVs to accommodate transportation requests, offering point-to-point services with ride sharing. We focus on the two major problems of the system: scheduling and admission control. The former is to configure the most economical schedules and routes for the AVs to satisfy the admissible requests while the latter is to determine the set of admissible requests among all requests to produce maximum profit. The scheduling problem is formulated as a mixed-integer linear program and the admission control problem is cast as a bilevel optimization, which embeds the scheduling problem as the major constraint. By utilizing the analytical properties of the problem, we develop an effective genetic-algorithm-based method to tackle the admission control problem. We validate the performance of the algorithm with real-world transportation service data.

\end{abstract}


\IEEEpeerreviewmaketitle

\section{Introduction}

\IEEEPARstart{H}{uman} 
mobility is largely supported by public \textcolor{black}{transport}. Many  people rely on public \textcolor{black}{transport} to move from one place to another when the destinations of their journeys are not within walkable distances. 
To transform a city with limited room for large-scale infrastructure into a smart city, its public transportation system may \textcolor{black}{need to} be further upgraded mainly from the existing road networks.  
Representatives of road-based public transport are buses and taxis, each \textcolor{black}{type} of which has its pros and cons. In general, buses follow fixed routes offering shared ride so that more passengers can be served on each \textcolor{black}{single journey}. On the other hand,  taxis offer private services and  run on flexible dedicated routes based on the passengers' requests. Nevertheless, no single one type can support high throughput and flexibility at the same time.
The efficiency and capacity of the whole public transportation system may be enhanced if there exists a new public transport which can accommodate many people in a short period of time and concur high mobility. 
 It may maintain flexibility by offering point-to-point services while enhancing efficiency by supporting shared ride. 
Such kind of public transport requires several characteristics which may not be possessed by a typical public \textcolor{black}{transport}. To develop such a public transport, the vehicles need to cooperate to take up customers' requests instead of cruising around the city for random offers. To enhance the efficiency and cooperativeness, a control center can be employed to coordinate all the vehicles, manage all the service requests, and assign the vehicles to serve the requests. Moreover, the vehicles should follow the routes and carry out the travel plans instructed so as to achieve system-wise objectives.
Recently, autonomous vehicles (AVs) have been undergone active research and we can expect many AVs running on the roads in the near future. The AV is a good candidate possessing most of the requirements mentioned above. Hence AVs can be adopted to \textcolor{black}{construct} a new smart public transportation system with high efficiency and flexibility.

In this paper, we introduce an intelligent \textcolor{black}{AV-based} public transportation system. \textcolor{black}{It manages a fleet of AVs to accommodate transportation requests, offering point-to-point services with ride sharing.}  
We focus on two important problems in the system: \textit{scheduling} and \textit{admission control}. The former is about how to assign the designated vehicles to the admissible transportation requests, and when and where the vehicles should reach to provide services \textcolor{black}{with the lowest cost}. The latter is to determine the set of admissible requests among all requests to achieve maximum revenue.  
As a whole, the contributions of this paper include:
\begin{itemize}
	\item \textcolor{black}{proposing the} AV public transportation system;
	\item improving the model for scheduling proposed in \cite{confversion}, such that the formulation developed in this paper can now support both directed  and undirected graphs;
	\item developing distributed scheduling;
	\item formulating the admission control problem;
	\item introducing the concept of admissibility and the related analytical results;
	\item \textcolor{black}{designing} an effective method to solve the admission control problem; and
	\item validating the performance of the solution method with real-world transportation service data.
\end{itemize}
The rest of this paper is organized as follows. \textcolor{black}{Related work is given in Section \ref{sec:related} and we} present various system components and their operations in Section \ref{sec:model}. The scheduling problem is discussed in Section \ref{sec:scheduling}. In Section \ref{sec:admission}, we formulate the admission control problem and provide the related analytical results. We propose a genetic-algorithm-based solution method for admission control and develop distributed scheduling in Section \ref{sec:algorithm}. Section \ref{sec:simulation} evaluates the system performance with real-world transportation service data. Finally we conclude this paper in Section \ref{sec:conclusion}.

\section{Related Work} \label{sec:related}

The concept of AVs was raised in the 1920's and the research thereof has started for more than thirty years. 
An AV is equipped with many sensors, 
which provide the vehicle with full sensing ability so as to adapt to the neighborhood environment and realize fully automated control.
In 2007, the DARPA Urban Challenge boosted the awareness of AVs capable of being driven in traffic and performing complex maneuvers \cite{darpa}. In 2010, VisLab carried out the experiment that several driverless vehicles successfully traveled 13,000 km from Italy to China \cite{italychina}. Google demonstrated an AV prototype in 2011 \cite{google}. 
By the end of 2013, several states in the \textcolor{black}{United States}, including Nevada, Florida, California, and Michigan, had passed the law to allow AVs running on public roads \cite{AVlegal}.
The first self-driving shuttle on sale was from NAVIA \cite{NAVIA}.
Other automotive manufacturers, like \textcolor{black}{Mercedes-Benz} \cite{benz}, BMW, and Audi \cite{BMW},  have invested in self-driving technologies and include  AVs in their  production plans. 

Most research \textcolor{black}{work} on AVs mainly focused on the control and communication aspects. Mladenovic and Abbas \cite{controlFramework} proposed a self-organizing and cooperative control framework for distributed vehicle intelligence. 
Hu \textit{et al.} \cite{assignment}  studied lane assignment strategies for connected AVs and proposed a lane changing maneuver to balance the tradeoff between efficiency and safety. 
Petrov and Nashashibi \cite{feedback} developed a feedback controller for autonomous overtaking without utilizing roadway marking and inter-vehicle communication. 
Li \textit{et al.} \cite {multilevel} presented a multi-level fusion-based road detection system for driverless vehicle navigation to ensure safety in various road conditions. 
All these show that AV is a promising technology with the support from governments, high-tech companies, and car manufacturers.

\textcolor{black}{Vehicles can communicate with each other and fixed infrastructure via various vehicular wireless communication techniques \cite{vehcomm}. Nowadays vehicular communications are mostly deployed over satellite, cellular networks, and vehicular ad-hoc networks (VANETs) \cite{vehcomm}.}
VANET is a mobile ad-hoc network where vehicles act as the mobile nodes \cite{vanet} and it can improve the communication capacity and organization of AVs constituting an intelligent transportation system. 
Furda \textit{et al.} \cite{wireless} introduced a wireless communication framework for driverless vehicles. 
It facilitated  vehicle-to-vehicle and vehicle-to-infrastructure communications and improved the safety and efficiency of  vehicles. 
Alsabaan \textit{et al.} \cite{v2v} made use of traffic light signals and vehicle-to-vehicle (V2V) communications to help vehicles adapt their speeds and avoid unnecessary stop, acceleration, and excessive speed. 
Gomes \textit{et al.} \cite{video} designed a driver-assistance system which allowed a vehicle to collect real-time camera images from other vehicles in the neighborhood over V2V communications. 
In this way, AVs become connected and can communicate with the control center.

Shareability of taxi services has been studied recently. Santi \textit{et al.} \cite{taxipooling} investigated the tradeoff between passenger inconvenience and collective benefits of sharing and concluded that a small increase in discomfort could induce the significant \textcolor{black}{benefits} of less congestion, less running costs, less split fares, less polluted, and cleaner environment.
Ma \textit{et al.} proposed a taxi ridesharing system called T-Share in  \cite{t-share}, where the dynamic taxi ridesharing problem was studied. For a dataset of taxi services in Beijing, it showed that 25\% additional taxi users could be served with saving of 13\% of total travel distance.
These studies confirmed that ridesharing is beneficial but they mostly focused on taxi services.
In this paper, we focus on \textcolor{black}{AVs}, which \textcolor{black}{have} a key intrinsic property hardly found in the standard \textcolor{black}{taxis}: the direct control of vehicles does not involve any human factors. In other words, AVs can completely follow the instructions from the control center in the sense that they neither undertake any unassigned requests nor reject any assigned requests. We can see that AVs can fully cooperate to achieve the system objective but it may not be the case for \textcolor{black}{human-driving} taxis.

\begin{table}[!t]
\renewcommand{\arraystretch}{1.3}
\caption{Contributions to the system.}
\label{tab:related}
\centering
\textcolor{black}{
\begin{tabular}{p{1.5cm}|p{1.8cm}|p{3.2cm}|c }
\hline\hline
Technology/	& \multirow{2}{*}{Example} 	&	\multirow{2}{*}{Contributions}	& \multirow{2}{*}{Ref.}	\\ 
feature			&														&													&												\\\hline
\multirow{5}{*}{Hardware}	& VisLab 											&	Demonstrate the feasibility of AVs	& \cite{italychina}\\
					& Google 											&	Show the confidence of the industry in AVs			& \cite{google}\\
					& Mercedes-Benz, BMW, Audi, NAVIA					&	Guarantee supply of AVs for the system 								& \cite{benz,BMW,NAVIA}				\\ \hline
\multirow{6}{*}{Software}	&	Mladenovic \& Abbas 				&	Enhance self-organizing and cooperative control of AVs												&\cite{controlFramework}\\
					& Hu \textit{et al.} 					&	Balance the efficiency and safety of AVs													&\cite{assignment}			\\
					& Petrov \& Nashashibi 				&	Enhance self-control  of AVs 													&	\cite{feedback}				\\
					& Li \textit{et al.} 					&	Improve safety of AVs												&\cite {multilevel}			\\ \hline
\multirow{3}{*}{Law}				& Nevada, Florida, California, and Michigan	&	Demonstrate the support of governments	& \cite{AVlegal}\\ \hline
				& Cottingham						& 	Introduce the vehicular wireless communications available to be used in the system	& \cite{vehcomm}				\\
								& Dahiya \& Chauhan			&		Improve the communication capacity and organization of AVs										& \cite{vanet}					\\
Communications					& Furda \textit{et al.}	&	Enhance the communications between AVs and the control center												& \cite{wireless}				\\
								& Alsabaan \textit{et al.}	&	Improve the comfort of AVs										& \cite{v2v} \\
								& Gomes \textit{et al.}		&	Collect data for the system to estimate traffic conditions											&	\cite{video}\\ \hline
\multirow{2}{*}{Ridesharing}			& Santi \textit{et al.}	&		\multirow{2}{3.2cm}{Confirm the ridesharing functionality of the system}											& \cite{taxipooling}\\
								& Ma \textit{et al.}		&													& \cite{t-share}\\\hline
\multirow{2}{1.5cm}{AV public transportation system}		& Lam \textit{et al.}	& Provide a proof of concept	& \cite{confversion} \\
																										& Lam \textit{et al.}	& Investigate the scheduling and admission control problems	& This work\\
\hline\hline
\end{tabular}
}
\end{table}

The AV public transportation system is uniquely designed and it can help improve the capacity and flexibility of the future transportation system. To further demonstrate its feasibility, we show how the existing work discussed above may contribute to the system in Table \ref{tab:related}.

The scheduling problem has been introduced in \cite{confversion} and it can be considered as a variant of the Dial-A-Ride Problem (DARP) \cite{DARP}.  However, in our AV scheduling problem, we allow modifying the previously assigned but not yet served requests at desirable times to achieve system-wise performance goal. When the system evolves, the AVs appear at different locations at different time instants. It may happen that a particular request can be better served by a different AV at different times. \textcolor{black}{Consider an example with two AVs, I and II. At a paricular time, AV-I  is in the neighborhood of a location while AV-II is not. A request originated from this location may be better served by AV-I. After some time, AV-I may have gone away but AV-II may have come into the neighborhood. Then the request may be better served by AV-II instead.} As the AVs are connected \textcolor{black}{through appropriate vehicular communication technologies}, the schedules of AVs can be revised from time to time. We \textcolor{black}{consider} this in our formulation \textcolor{black}{making} our scheduling problem different from DARP.
\textcolor{black}{As the system involves a number of AVs, determining their schedules in a distributed manner can undoubtedly speed up the process. Distributed scheduling has been advanced in many engineering disciplines, e.g.,  communication networks \cite{distributedscheduling1,distributedscheduling2}. As a new system, we will dedicatedly design a distributed  methodology for the scheduling thereof.}

\textcolor{black}{Admission control generally refers to a validation process in communication systems for quality-of-service assurance. It determines which new connection or service request can be granted with resources for subsequent operations.  For example, \cite{admission1} designed an admission control mehanism to add or drop session requests in 4G wireless networks and \cite{admission2} discussed various admission control algorithms for multi-service IP networks. We adopt this idea in the transportation system and design an admission control mechanism to differentiate the transportation service requests for maximizing the total profit.}
\textcolor{black}{There are many methods to facilitate admission control. Genetic Algorithm (GA) is one of them and it has been successfully utilized to design admission control mechanisms, e.g., \cite{admission_algo1} and \cite{admission_algo2}. Based on the special formulation of the admission control problem (to be discussed in Section \ref{sec:admission}), we will also adopt GA to solve the problem.}
\section{System Model}\label{sec:model}

In this section, we design the architecture for the system which can manage a fleet of AVs to serve customers for transportation services. In the following, we first introduce the system components and then describe the operations  characterizing their interactions.

\subsection{System Components}

\subsubsection{Network Structure}
A graph is employed to model the region being served by the system. It characterizes the locations and the road connections necessarily to describe movements of the AVs, origins and destinations of the service requests, and other required facilities. It is  a directed graph denoted by $G(\mathcal{V},\mathcal{E})$, where $\mathcal{V}$ is a set of locations and $\mathcal{E}$ refers to the road segments connecting the locations so that we can completely describe the routes of AVs \textcolor{black}{with} $G$.  
For $i,j\in \mathcal{V}$, each edge $(i,j)\in\mathcal{E}$ is associated with an operational cost $c_{ij}$ and a travel time $t_{ij}$, which is an estimation of time for an AV to traverse from $i$ to $j$ based on historical data. Depended on the system objective, $c_{ij}$ typically represents the distance of the road segment $(i,j)$ as the operational cost of AVs  is usually measured by the fuel consumption which is in turn characterized by the travel distance. If the system aims to optimize the total service duration, we can set $c_{ij}=t_{ij}$ for all $(i,j)$'s. We allow $c_{ij}\neq c_{ji}$ and $t_{ij}\neq t_{ji}$ to account for the asymmetry of road segments. Moreover, refuel stations are located in some locations specified by $\tilde{\mathcal{V}}\subset \mathcal{V}$ and each AV  ends its journey at any one of these refuel stations (reasons explained in Section \ref{subsubsec:breakdown}). Based on the nature of the AVs, $\tilde{\mathcal{V}}\subset \mathcal{V}$ will be the locations of charging (gas) stations if the AVs are electric (conventional) vehicles. For the case of electric vehicles, $\tilde{\mathcal{V}}\subset \mathcal{V}$ can be determined based on the charging demand and the connectivity of the charging station network according to \cite{EVCPP}.

\subsubsection{Transportation Requests}
Customers request services in the form of transportation requests, which are collectively denoted by $\mathcal{R}$. Each $r\in\mathcal{R}$ is represented by the 5-tuple $\langle s_r, d_r, T_r, [e_r, l_r], q_r \rangle$. $s_r\in \mathcal{V}$ and $d_r\in \mathcal{V}$ represent the customer pickup and dropoff locations, respectively. $T_r$ is the maximum ride time, an exceedance of which will lead to customer dissatisfaction. $[e_r, l_r]$ refers to the service starting time window, where  $e_r$ and $l_r$ are the earliest and latest service starting times, respectively. $q_r$ stands for the number of seats \textcolor{black}{needed in the request $r$}.

\subsubsection{Vehicles}
The system coordinates a fleet of AVs denoted by $\mathcal{K}$. Each $k\in \mathcal{K}$ is represented by the 5-tuple $\langle a_k,t_k^0,\tilde{T}_k, Q_k, \mathcal{R}_k\rangle$. $a_k\in\mathcal{V}$ is the first location where $k$ will visit from the current position of $k$ while $t^0_k$ is the time required to reach $a_k$ from its current position.  It is possible that, at the time of scheduling, the AV is in the middle of a road segment heading to $a_k$. $a_k$ and $t^0_k$ can be easily estimated by submitting its current position to the system. $\tilde{T}_k$ denotes the maximum remaining operation time that $k$ can continue to provide services without refueling.\footnote{\textcolor{black}{The maximum remaining operation time of $k$ can be converted from its corresponding remaining fuel level.}} $Q_k$ is the passenger capacity that $k$ can accommodate simultaneously. $\mathcal{R}_k= \tilde{\mathcal{R}}_k\cup \overline{\mathcal{R}}_k \in \mathcal{R}$ is the set of requests previously assigned to $k$. $\mathcal{R}_k$ can be further categorized into two types; $\tilde{\mathcal{R}}_k$ contains those  currently being served by $k$ while $\overline{\mathcal{R}}_k$ was assigned to $k$ at a previous schedule but the services have not been implemented yet. For the former, some seats have already been taken by the customers from $\tilde{\mathcal{R}}_k$. On the contrary, seats have only been reserved but no actual seats \textcolor{black}{have been} taken from $\overline{\mathcal{R}}_k$. We will handle $\tilde{\mathcal{R}}_k$ and $\overline{\mathcal{R}}_k$ differently when performing scheduling in Section \ref{sec:scheduling}.

Without loss of generality, we assume that the number of seats required in any request is no larger than the capacity of any vehicle, i.e., 
\begin{align}
q_r\leq Q_k,\forall r\in\mathcal{R}, k\in\mathcal{K}.  \label{capCondition}
\end{align}
We can always split those requests violating \eqref{capCondition} into multiple requests so that this condition always holds.

\subsection{Operations} \label{subsec:op}

\begin{figure}[!t]
\hspace{-0.3cm}
\includegraphics[width=3.8in]{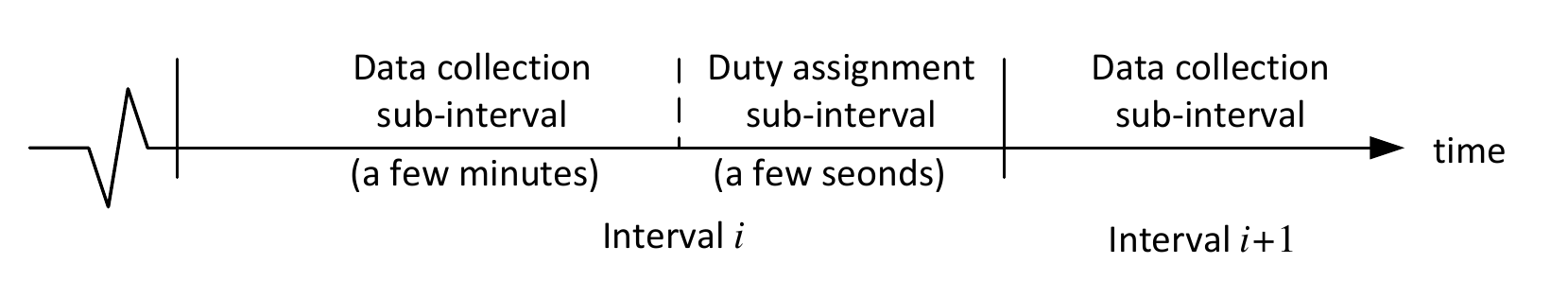} \vspace{-0.5cm}
\caption{Operating intervals in the system.} 
\label{fig:interval}
\end{figure}

\begin{figure}[!t]
\centering
\includegraphics[width=3.2in]{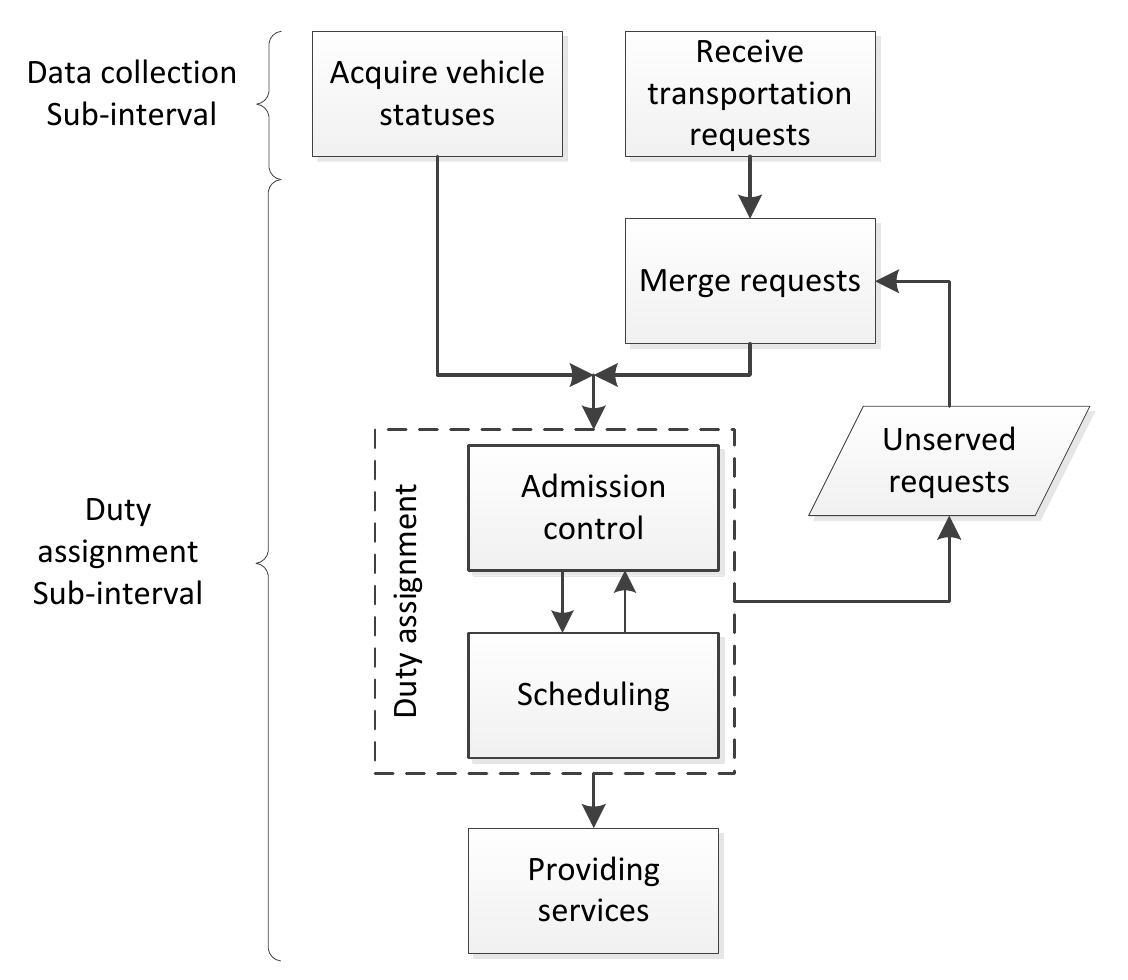} \vspace{-0.5cm}
\caption{Operation flow of the system.} 
\label{fig:operations}
\end{figure}

The system is managed and operated by a control center whose main duties  are to collect all the required information and assign the AVs to serve the transportation requests. 
The system operates in a fixed time interval basis and each time interval is divided into data collection and duty assignment sub-intervals (see Fig. \ref{fig:interval}). In each interval, the control center first collects transportation requests and vehicle statuses in the data collection sub-interval. Then the AVs are assigned to \textcolor{black}{serve} the transportation requests in the duty assignment sub-interval.
On one hand, the duration of each interval should be long enough such that the communication delays will not result in  any data missing from the customers and vehicles for scheduling. On the other hand, it should be short enough such that the collected data can reflect the current situation happening in that interval.
In practice, the data collection sub-interval is longer than the duty assignment one. The former may last for a few minutes while the latter may takes a few seconds.

Fig. \ref{fig:operations} illustrates the operation flow of the system \textcolor{black}{with respect to an operating interval}. As powered by \textcolor{black}{various wireless} vehicular \textcolor{black}{communication} technologies, all AVs are connected and can communicate with the control center \textcolor{black}{instantaneously}. In this way, the control center \textcolor{black}{can collect} the necessary vehicle statuses, e.g., current locations of AVs, confirmation of serving requests, traffic congestion information, etc.\textcolor{black}{, in the data collection sub-interval.} Customers \textcolor{black}{can} also submit their requests to the control center by any appropriate means, e.g., phone calls, mobile apps, etc. After the data collection sub-interval, all the data required to perform duty assignment are ready \textcolor{black}{at the control center}.

In the duty assignment sub-interval, the control center processes the collected data and computes the duty assignment.
There may exist some unattended requests incurred from some previous intervals because of their unsuitability in the \textcolor{black}{previous} system conditions. They are merged with the newly \textcolor{black}{submitted} requests and then all these requests are considered \textit{en masse}. 
The duty assignment further consists of two processes: admission control and scheduling. Admission control checks all the outstanding requests and determines which requests are going to be admitted in the current interval. The unadmitted requests will be reserved for consideration in the next interval again.
Any invalid or inappropriate requests are also permanently excluded in the admission control \textcolor{black}{process}. We compute the travel schedules of the AVs to serve the admitted requests in the scheduling process. 
If a vehicle is assigned with a request, its schedule settled by the control center needs to satisfy the following requirements:
\subsubsection{Complete route specification} \label{subsubsec:breakdown}
Since the vehicle is unmanned, we need to specify the \textit{exact} route so that the vehicle can follow the route to pick the passengers of the assigned requests up and to drop them off at the required destinations. Moreover, the route should be short enough so that it has sufficient fuel to complete the route. The vehicle should end up at a refuel station to avoid breaking down in the middle of any road segments. This can guarantee that the vehicle must be able to refuel after completing all the assigned services.
\subsubsection{Time constraints}
The vehicle should be able to pick the passengers up at \textcolor{black}{a} time within the service starting time window specified in the request. Moreover, the actual ride time should be no longer than the maximum value stated in the request.
\subsubsection{Capacity constraints}
When the vehicle arrives at the pickup location, there should always be enough free seats available to accommodate all the passengers of the request.

Admission control and scheduling are inter-related and we will discuss their details in the subsequent sections.
After determining the result, the \textcolor{black}{control} center then distributes the assignments to the corresponding AVs, which provide services to the customers.

\section{Scheduling}\label{sec:scheduling}
Scheduling involves determining the following:
\begin{itemize}
	\item the assignment of AVs to the requests;
	\item the routes of AVs to accomplish the assigned requests; and
	\item the times by which the AVs should reach particular locations.
\end{itemize}
\textcolor{black}{Here we assume that all requests being scheduled are admissible, where the admittability of a request is handled by admission control.} \textcolor{black}{Thus all requests will be served by appropriate vehicles after scheduling.} When discussing admission control in Section \ref{sec:admission}, we will explain the relationship between admission control and scheduling. 

To facilitate scheduling, we assume that all vehicles are connected and can communicate with the control center with reasonably short delays.
This ensures that no apparent changes in positions happen to the AVs  
in each interval given in Fig. \ref{fig:interval}. With the support of \textcolor{black}{modern} advanced communication technologies, this assumption can go through.
In our model, we require that the computation of scheduling can be done in a short period of time. This ensures the validity of the traffic data when the vehicles traverse along their assigned routes. There are basically two types of traffic data: the distances and travel times of road segments. The former is time-invariant while the latter usually changes gradually. In other words, significant changes in travel times only take place in a timespan much longer than the time interval.

\subsection{Preprocessing} \label{subsec:preprocessing}
We schedule  the AVs to accomplish the transportation requests to achieve  the minimum total operational cost in terms of fuel costs, which are in turn measured by the total distance traveled. The distance between any pair of locations is invariant and we transform $G(\mathcal{V},\mathcal{E})$ to $G'(\mathcal{V}',\mathcal{E}')$ with any shortest path algorithm, e.g., Dijkstra's algorithm \cite{dijkstra}, where $\mathcal{V}'\subset \mathcal{V}$ is the set of locations at which we need to determine the arrival times of the assigned AVs in order to configure their travel schedules. $\mathcal{V}'$ includes the first locations visited by all the vehicles (i.e., $a_k$'s), the sources and destinations of the requests (i.e., $s_r$'s and $d_r$'s, respectively), and the locations of the refuel stations (i.e., $i\in \tilde{\mathcal{V}}$). $\mathcal{E}'$ is defined as $\{(i,j)|i,j\in\mathcal{V}' \}$ such that there exists a shortest path from $i\in \mathcal{V}$ to $j\in \mathcal{V}$ in $G$. For $(i,j)\in \mathcal{E}'$, the associated $c_{ij}$ and $t_{ij}$ are the sums of costs and times, respectively, of all the edges constituting the corresponding shortest path in $G$. 
In the subsequent computation, we focus on $G'(\mathcal{V}',\mathcal{E}')$ instead of $G(\mathcal{V},\mathcal{E})$.
The reasons why we adopt this transformation are two-fold: 
First, the number of variables needed in the formulation can be dramatically reduced. The set $\mathcal{V}\setminus \mathcal{V}'$ are not important as all conditions confining to the locations  specified by the vehicles and requests are restricted to $\mathcal{V}'$ only. In this way, the efficiency of solving the scheduling problem can be improved significantly. 
Second, this can improve the flexibility of the schedules. Consider that AV $k$ goes from vertices 1 to 4 and there exist two paths connecting them as, Path 1: $1\rightarrow 2\rightarrow 4$, and Path 2: $1\rightarrow 3\rightarrow 4$.  Suppose that vertices 1 and 4 belong to $\mathcal{V}'$ but vertices 2 and 3 do not. To satisfy the requirements imposed on $k$, we need to determine the times by which $k$ should arrive at vertices 1 and 4 only, i.e., $t_1^k$ and $t_4^k$.
If vertices 2 and 3 are also included in the formulation and Path 1 is finally chosen, $t_2^k$ will be specified by solving the scheduling problem and thus $k$ needs to arrive at the vertices by $t_1^k$, $t_2^k$, and $t_4^k$, respectively. If not, only $t_1^k$ and $t_4^k$ are specified and we can give flexibility to $k$ of arriving at vertex 2. $t_2^k$ can be any time between $t_1^k$ and $t_4^k$ as long as the required travel times spent on $(1,2)$ and $(2,4)$ have been considered. This flexibility gives room for $k$ to respond to any instantaneous traffic incidents which may disturb its original travel plan. This also allows $k$ to change to Path 2, if needed, without altering the original travel plan.

\textcolor{black}{Note that the preprocessing step can be skipped if the scheduling problem constructed directly from $G(\mathcal{V},\mathcal{E})$ can be solved efficiently. However, if the preprocessing is required to simplifiy the scheduling problem, it can be considered as a number of result lookups. As $c_{ij}$'s generally refer to the travel distances which are invariant, the results of the shortest path computations are also invariant. In fact, before the system operates, we can first compute the shortest path for every pair of locations in $\mathcal{V}$. When the preprocessing is triggered in an interval, we just need to look up the pre-computed shortest path results. Hence, the time cost of preprocessing can be considered negligibly small.}

\subsection{Problem Formulation}
We formulate the scheduling problem based on $G'(\mathcal{V}',\mathcal{E}')$.
The given data for the problem parameters include the graph $G'(\mathcal{V}',\mathcal{E}')$ with costs $c_{ij}$'s and travel times $t_{ij}$'s, the set of transportation requests $\mathcal{R}$, and the set of AVs $\mathcal{K}$. We define several variables for the problem. Binary variables $x_{ij}^k$'s are used to indicate which connections will be traversed by the vehicles, as
\begin{align*}
x_{ij}^k=\left\{
	\begin{array}{ll}
		1 & \text{if vehicle $k$ traverses $(i,j)$,}\\
		0 & \text{otherwise.}
	\end{array}
\right.
\end{align*}
We define binary variables $y_r^k$'s for the assignment of the vehicles to the requests, as
\begin{align*}
y_{r}^k=\left\{
	\begin{array}{ll}
		1 & \text{if vehicle $k$ is assigned to request $r$,}\\
		0 & \text{otherwise.}
	\end{array}
\right.
\end{align*}
For $i\in \tilde{\mathcal{V}}$, binary variables $g_i^k$'s are utilized to indicate the refuel stations at which the vehicles end their routes, as 
\begin{align*}
g_{i}^k=\left\{
	\begin{array}{ll}
		1 & \text{if vehicle $k$ ends its route at vertex $i\in \tilde{\mathcal{V}}$,}\\
		0 & \text{otherwise.}
	\end{array}
\right.
\end{align*}
We need to specify the times and occupancy conditions at various locations along the routes. Let $t_i^k$ be the time by which $k$ should arrive at vertex $i$ and $f_i^k$ be the number of passengers in $k$ right before it leaves  $i$. 


We aim to construct economical schedules for the AVs and thus we minimize the total operational cost with the objective function as
\begin{align}
\sum_{i,j\in \mathcal{V}, k\in \mathcal{K}} c_{ij}x_{ij}^k. \label{objective}
\end{align}
We define a set of constraints to confine the scope of the variables so that the requirements discussed in Section \ref{subsec:op} are satisfied.
Each transportation request can only be served once and thus we have
\begin{align}
\sum_{k\in\mathcal{K}}y_r^{k} = 1, \forall r\in\mathcal{R}. \label{requestassignment}
\end{align}
Each AV will end at one of the refuel stations if it is assigned to a request. This is specified by
\begin{align}
\sum_{i\in\tilde{\mathcal{V}}}g_i^{k} \leq 1, \forall k\in\mathcal{K}. \label{gasAssignment}
\end{align}
If AV $k$ is not assigned to any request, we do not need to determine a path for $k$ so as the final stopping refuel station for $k$. Thus it is possible to have $\sum_{i\in\tilde{\mathcal{V}}}g_i^{k}=0$ \textcolor{black}{for some $k$}.

Let $\mathcal{N}^+(i)$ and $\mathcal{N}^-(i)$ be the sets of incoming and outgoing neighbors of vertex $i$, i.e., $\mathcal{N}^+(i)=\{j\in\mathcal{V}'|(j,i)\in\mathcal{E}'\}$ and $\mathcal{N}^-(i)=\{j\in\mathcal{V}'|(i,j)\in\mathcal{E}'\}$. We model a path with a network flow model. A path starting at $a_k$ and ending at $i\in\tilde{\mathcal{V}}$ can be defined with the following: 
\begin{align}
0\leq \sum_{i\in \mathcal{N}^-(a_k)}x_{a_ki}^k - \sum_{i\in \mathcal{N}^+(a_k)}x_{ia_k}^k \leq \sum_r{y_r^k},\forall k\in \mathcal{K}, \label{startflow}\\
0\leq \sum_{j\in \mathcal{N}^+(i)}x_{ji}^k - \sum_{j\in \mathcal{N}^-(i)}x_{ij}^k \leq g_i^k,\forall i\in \tilde{\mathcal{V}}, k\in \mathcal{K}, \label{endflow}\\
\sum_{j\in \mathcal{N}^+(i)}x_{ji}^k = \sum_{j\in \mathcal{N}^-(i)}x_{ij}^k, \forall i\in\mathcal{V}'\setminus \tilde{\mathcal{V}}\cup\{a_k|k\in \mathcal{K}\}. \label{otherflow}
\end{align}
Eq. \eqref{startflow} defines for the starting vertex of $k$, where a starting vertex has one unit of net outgoing flow.
$\sum_r{y_r^k}$ specifies if a path needs to be defined for $k$. If there are no requests assigned to $k$, $\sum_r{y_r^k}$ becomes zero and $a_k$ is not the starting vertex of any paths \textcolor{black}{for $k$}.
Similarly, \eqref{endflow} defines for the destination vertex of $k$ and the exact vertex $i$ ended by $k$ is indicated by $g_i^k$. If $k$ ends at $i\in \tilde{\mathcal{V}}$, \eqref{endflow} will allow $i$ to have one unit of net incoming flow for $k$. 
For other vertices, \eqref{otherflow} sets the conversation of flow by equalizing the corresponding incoming and outgoing flows.

If request $r$ is assigned to vehicle $k$, $k$ needs to pass through the pickup location $s_r$ of $r$. It is equivalent to having positive outgoing flow for $k$ at $s_r$  as
\begin{align}
\sum_{i\in\mathcal{N}^-(s_r)}x_{s_ri}^k \geq y_r^k, \forall r\in\mathcal{R},k\in\mathcal{K}. \label{possourceflow}
\end{align}
Similarly, $k$ needs to pass through the dropoff point $d_r$ of request $r$ when $r$ is served by $k$. This requires positive incoming flow for $k$ at $d_r$ as
\begin{align}
\sum_{i\in\mathcal{N}^+(d_r)}x_{id_r}^k \geq y_r^k, \forall r\in\mathcal{R},k\in\mathcal{K}. \label{posdesflow}
\end{align}
Note that specifying incoming flow for $s_r$ is not sufficient as it is possible to have zero incoming flow when $k$ begins its path at $s_r$ exactly. Similarly, it is not sufficient to specify outgoing flow for $d_r$ as it is possible to have zero outgoing flow when $k$ ends its path at $d_r$.


No matter where vehicle $k$ goes, it cannot travel continuously longer than its \textcolor{black}{operational time} limit specified by $\tilde{T}_k$. Moreover, it needs to take at least $t_k^0$ in order to reach the initial vertex of its path. Hence we have
\begin{align}
 t_k^0\leq t_i^k \leq \tilde{T}_k, \forall i\in \mathcal{V}', k\in \mathcal{K}. \label{timescope}
\end{align}

Let $M$ be a sufficiently large positive number.
When vehicle $k$ traverses edge $(i,j)$, the time at $j$ should be larger than or equal to the time at $i$ together with the travel time on $(i,j)$, i.e., $t_{ij}$. This can be specified by
\begin{align}
	t_j^k \geq t_i^k + t_{ij} - M(1-x_{ij}^k), \forall k\in\mathcal{K}, i,j\in\mathcal{V}'. \label{updatet}
\end{align}

When vehicle $k$ is assigned to request $r$, the actual ride time to reach $d_r$ from $s_r$ should be no larger than the maximum ride time $T_r$ specified by $r$, i.e.,
\begin{align}
	t_{d_r}^k - t_{s_r}^k \leq T_r + M(1-y_r^k), \forall r\in \mathcal{R},k\in \mathcal{K}.
\end{align}

If request $r$ is served by vehicle $k$, $k$ should arrive at $s_r$ within the service starting time window $[e_r,l_r]$ specified by $r$. This can be expressed as
\begin{align}
e_r-M(1-y_r^k)\leq t_{s_r}^k\leq l_r+M(1-y_r^k),\forall r\in \mathcal{R},k\in \mathcal{K}. \label{window}
\end{align}

Passengers being served occupy seats and the capacity limits of all vehicles should be satisfied at all times. So we have
\begin{align}
0\leq f_i^k \leq Q_k, \forall i\in\mathcal{V}', k\in\mathcal{K}. \label{capCon}
\end{align}

At $a_k$, some passengers induced from $\tilde{\mathcal{R}}_k$ may get off $k$ and new passengers may get on $k$ from other requests. The occupancy conditions of the AVs at their initial vertices $a_k$'s are given by
\begin{align}
	f_{a_k}^k \geq \sum_{r|s_r=a_k} q_ry_r^k - \sum_{r|d_r=a_k} q_ry_r^k, \forall k\in\mathcal{K}. \label{sourceCap}
\end{align}
When $k$ traverses from $i$ to $j$ along $(i,j)$, vertex $j$ may be the pickup locations of some requests and dropoff locations of some other requests. The \textcolor{black}{relationship between the} occupancy \textcolor{black}{conditions} of AV $k$ at \textcolor{black}{$i$ and} $j$ \textcolor{black}{can be}  specified  as
\begin{align}
	f_{j}^k \geq f_{i}^k - M(1-x_{ij}^k)+ \sum_{r|s_r=a_k} q_ry_r^k - \sum_{r|d_r=a_k} q_ry_r^k, \\ \nonumber
	\forall i,j\in\mathcal{V}', k\in\mathcal{K}. \label{desCap}
\end{align}

When an AV reaches a refuel station, all requests assigned to it should have been settled and no passenger should be accompanied to the end of the route. This is described by
\begin{equation}
	f_i^k\leq M(1-g_i^k), \forall i\in \tilde{V},k\in \mathcal{K}.  \label{refuelCap}
\end{equation}


Recall that there are two kinds of requests which have already been assigned to the AVs before the current scheduling interval, i.e., $\mathcal{R}_k=\tilde{\mathcal{R}}_k\cup \overline{\mathcal{R}}_k$. As a (nearly) real-time application, with updated information, we may further improve the system performance by revising the already assigned requests. For those requests currently being served, e.g., $r\in\tilde{\mathcal{R}}_k$ with the passengers sitting in $k$, we can consider those $r$'s as ``new'' requests starting the service at the the starting node $a_k$ by setting $s_r=a_k$ and affirming $y_r^k=1$. As $k$ has been serving $r$ by following a previously determined schedule, we can update its $T_r$ by shortening the elapsed time. The service starting time window is no longer important and thus  we set $e_r=-\infty$ and $l_r=+\infty$. There is no change to $q_r$. For those requests $\overline{\mathcal{R}}_k$'s which have been previously assigned to $k$ but not yet been served, we may reschedule $r\in\overline{\mathcal{R}}_k$ with other AVs if it can result in lower cost. As the passengers do not concern about which vehicle would eventually provide the service, it may be more efficient  to re-allocate those $r$'s in $\overline{\mathcal{R}}_k$  to other more appropriate vehicles with lower operational cost. This enhances the flexibility of the system.
As a whole, the scheduling problem is defined as
\begin{problem}[Scheduling] \label{schedulingproblem}
\begin{align*}
\text{minimize}\quad 	& \eqref{objective}\\
\text{subject to}\quad & \eqref{requestassignment}-\eqref{refuelCap}\\
\text{over}\quad & x_{ij}^k\in\{0,1\}, y_r^k\in\{0,1\}, g_l^k\in\{0,1\}, t_i^k\in\mathbb{R}^+,\\
& f_i^k\in\mathbb{Z}^+,  \forall i,j\in \mathcal{V}', l\in \tilde{\mathcal{V}}, r\in\mathcal{R},k\in\mathcal{K}.
\end{align*}
\end{problem}
Problem \ref{schedulingproblem} has a linear objective function and linear equality and inequality constraints. Some of its variables are binary while the rest are real. Thus the scheduling problem is a mixed-integer linear program (MILP).
Although the preprocessing step discussed in Section \ref{subsec:preprocessing} helps simplify the problem, the numbers of variables and constraints also grow with the sizes of $\mathcal{R}$ and $\mathcal{K}$. 
\textcolor{black}{As} those invalid requests have been removed \textcolor{black}{by admission control (discussed in Section \ref{sec:admission})}, this MILP is always feasible and all requests must be served.
\textcolor{black}{As long as all $c_{ij}$'s are positive, the solution of Problem 1 does not result in zero cost and the schedule without serving any requests will never be a solution.}

\subsection{Complete Schedule Construction}
Since the vehicles are unmanned, we need to provide complete instructions about the paths and schedules so that they know \textcolor{black}{when and where they} should go in order to provide services to the customers.
Solving the MILP gives the solutions for $x_{ij}^k$'s, $y_r^k$'s, $b_k$'s , $t_i^k$'s, and $f_i^k$'s. As \textcolor{black}{being} binary variables, the results of  $x_{ij}^k$'s and $y_r^k$'s are \textcolor{black}{unambiguous}. The latter tells which vehicles are assigned to the requests. The former explains the route of each $k$ in $G'$ starting at $a_k$ and ending at one of the refuel stations. The paths determined in $G'$ in turn infer the corresponding complete routes in $G$. Recall that we have determined the shortest path from $i$ to $j$ in $G$ corresponding to the edge $(i,j)\in \mathcal{E}'$ . By inserting the shortest paths for every pair of adjacent vertices along the paths based on $G'$, the complete routes in $G$ can be derived accordingly.

Note that \eqref{timescope}--\eqref{window} define the scope of $t_i^k$'s in the form of inequality.
The resulting $t_i^k$'s make feasible time schedules but may not be specific enough \textcolor{black}{leading to ambiguity}. For example, if the arrival of $k$ at location $i$ at any moment in $[t',t'']$ is feasible, a reasonable way is to set $t_i^k=t'$ and this enhances the flexibility for  the later scheduling intervals. To construct the schedule of $k$, we examine the path computed from $x_{ij}^k$'s. For the first vertex, we set $t_{a_k}^k=t_k^0$. For any subsequent vertices, says from $i$ to $j$, we can add the travel time on edge $(i,j)$ to the settled time at $i$ to obtain the settled time at $j$, i.e., $t_j^k = t_i^k + t_{ij}$. If vertex $j$ induces a request, we need to fulfill its service starting time window and thus we have $t_j^k = \max\{t_i^k + t_{ij}, e_r\}$.

Similarly, \eqref{capCon}--\eqref{refuelCap} also confine the occupancies of the vehicles at various locations with inequalities. The exact seat conditions cannot be told from the resulting $f_i^k$'s. Usually, we only concern about the seat conditions at the customer pickup and dropoff points, i.e., $s_r$'s and $d_r$'s. We can examine the route computed from $x_{ij}^k$'s again and determine the occupancy conditions. For example, $k$ goes from $i$ to $j$ on $(i,j)$. If $j$ is the service starting location of request $r$, we add the number of seats required for $r$ to the occupancy of $k$ at $i$ to get its occupancy at $j$, i.e., $f_j^k = f_i^k + q_r$. If $j$ is a service destination location instead, we subtract the seats taken by $r$ from the occupancy of $k$ at $i$ to get its occupancy at $j$, as $f_j^k = f_i^k - q_r$.
\textcolor{black}{In this way, the complete schedules of the vehicles with duty assigned can be determined and the vehicles just need to follow the schedules to accomplish the services.}

\section{Admission Control}\label{sec:admission}

Recall that, in Section \ref{sec:scheduling}, all requests submitted for scheduling are assumed to be admissible and need to be served. In this section, we investigate \textcolor{black}{the admission control problem. We first formulate the problem and then study the variations in the presence of traffic congestion and no-show of passengers.}

\subsection{\textcolor{black}{Problem Formulation}}
\textcolor{black}{Admission control} is responsible for determining a set of requests suitable for scheduling. In other words, after admission control, we will produce a subset $\check{\mathcal{R}}\subset\mathcal{R}$ for subsequent scheduling, where $\mathcal{R}$ is the set of all available requests and $\check{\mathcal{R}}$ will be settled by appropriate AVs in scheduling. However, to judge if a particular request $r$ is admissible,  we need to check not only its feasibility but also its profitability, i.e.,  whether serving $r$ will induce a positive net profit. Determining the net profit from $r$ involves its induced cost, which is regulated through scheduling. Hence there is no clear precedence relationship between scheduling and admission control and these two processes should be considered simultaneously.

We can interpret the requests and AVs as the demand and supply of transportation services, respectively, and then the constraints of Problem \ref{schedulingproblem} define the scope of matching \textcolor{black}{between the demand and} supply. The constraints can be satisfied more easily with larger $\mathcal{K}$ and smaller $\mathcal{R}$. Practically, the size of $\mathcal{K}$ is generally fixed as the system would not suddenly employ more AVs into the fleet or many AVs become out of service all of a sudden. However, the requests submitted are absolutely external from the system; the system can neither forbid the customers from submitting requests nor modify the attributes in the requests to match the conditions of AVs. In fact, just a single inappropriate request (e.g., a request with very short tolerable ride time) can make Problem \ref{schedulingproblem} infeasible and the scheduling collapse. 
To avoid this, the system should perform admission control by screening out \textcolor{black}{any} inappropriate requests before undergoing the scheduling (see Fig. \ref{fig:operations}). Consider that entertaining a request results in revenue. Although the system cannot modify the submitted requests, it has the right to dismissing any requests by sacrificing the corresponding revenue. Admission control manipulates $\mathcal{R}$ with the following objectives: 
1) Produce a subset of requests $\check{\mathcal{R}}\subset\mathcal{R}$ so that the scheduling process can be performed, i.e., Problem \ref{schedulingproblem} is made feasible with $\check{\mathcal{R}}$;
2)  Maximize the profit incurred.

Consider that we admit $\check{\mathcal{R}}$ for scheduling with Problem \ref{schedulingproblem}, which \textcolor{black}{can be} re-written as
\begin{subequations}
\label{sim_scheduling}
\begin{align}
\text{minimize}\quad 	&\phi(\alpha) \\
\text{subject to}\quad 
& \alpha \in \mathcal{Z}(\check{\mathcal{R}}),
\end{align}
\end{subequations}
where $\alpha\triangleq \{x_{ij}^k\}\cup\{y_{r}^k\}\cup\{g_i^k\}\cup\{t_i^k\}\cup\{f_i^k\}$, $\phi(\alpha) \triangleq \sum_{i,j\in \mathcal{V}, k\in \mathcal{K}} c_{ij}x_{ij}^k$, and let $\mathcal{Z}(\check{\mathcal{R}})$ be the feasible region of Problem \ref{schedulingproblem} with respect to $\check{\mathcal{R}}$.
Let $\rho_r$ be the revenue made when admitting $r\in\mathcal{R}$ and define
\begin{align*}
z_{r}=\left\{
	\begin{array}{ll}
		1 & \text{if we admit $r\in\mathcal{R}$ for scheduling,}\\
		0 & \text{otherwise.}
	\end{array}
\right.
\end{align*}
We also define the admission function $\sigma(\mathcal{R},[z_r]_{r\in\mathcal{R}})$ which returns $\check{\mathcal{R}}\subset \mathcal{R}$ based on $z_r$ such that $r\in \check{\mathcal{R}}$ if $z_r=1$.
The total profit is the difference between the total revenue and total cost, i.e., $\sum_{r\in\mathcal{R}}\rho_r z_{r} - \phi(\alpha)$.
Then we formulate the admission control problem as
\begin{problem}[Admission Control] \label{acproblem}
\begin{subequations}
\label{admissioncontrol1}
\begin{align}
\text{maximize}\quad 	& \Phi(\mathcal{R},[z_r]_{r\in\mathcal{R}}) = \sum_{r\in\mathcal{R}}\rho_r z_{r} - \phi(\alpha) \label{bilevel_obj}\\
\text{subject to}\quad & \check{\mathcal{R}} = \sigma(\mathcal{R},[z_r]_{r\in\mathcal{R}}), \label{admission_func}\\
& z_r = 1, \forall r\in\mathcal{R}_k, k\in \mathcal{K}, \label{admittedReq}\\
& \alpha \in \arg\min \{\phi(\alpha):\alpha\in \mathcal{Z}(\check{\mathcal{R}})\}, \label{bilevel_lower}\\ 
\text{over}\quad & \alpha,\check{\mathcal{R}}\in\mathcal{R},z_r\in\{0,1\},\forall r\in\mathcal{R},
\end{align}
\end{subequations}
\end{problem}
where \eqref{admittedReq} ensures that those requests admitted in the previous operating intervals will still be admitted in the current interval.
We cast admission control as a bilevel optimization problem, which consists of an upper- and a lower-level optimization. $\Phi$ is the upper-level objective function with upper-level variables $\check{\mathcal{R}}$ and $z_r$'s. $\phi$ represents the lower-level objective function with lower-level variable $\alpha$. Eq. \eqref{bilevel_lower} is in fact \eqref{sim_scheduling}, and thus, we cast Problem \ref{schedulingproblem} as a constraint of Problem \ref{acproblem}.
The upper-level optimization is to manipulate the whole set of requests $\mathcal{R}$ and determine $\check{\mathcal{R}}$ such that $\check{\mathcal{R}}$ can maximize the total profit. The lower-level optimization is to schedule the AVs to serve the set of admissable requests $\check{\mathcal{R}}$ so that the retained cost is the lowest.
The two levels of optimization are inter-related; the upper level requires the result of the lower level, i.e., $\alpha$, in order to get $\check{\mathcal{R}}$, while the lower level needs the result from the upper level, i.e., $\check{\mathcal{R}}$, in order to output $\alpha$.
Note that if the upper level produces $\check{\mathcal{R}}$ which makes $\mathcal{Z}$ infeasible, the resulting $\alpha$ will return $+\infty$ for the objective function of \eqref{bilevel_lower}, which \textcolor{black}{will} in turn make the objective function \eqref{bilevel_obj} \textcolor{black}{retain} $-\infty$.

Bilevel optimization is in general difficult to solve. A bilevel problem  with a linear objective function and linear constraints is NP-hard \cite{bilevel}. As seen from \eqref{admissioncontrol1}, we are manipulating discrete variables in the problem. As classical methods for bilevel optimization usually assume smoothness or convexity \cite{bilevel_foundations}, those classical methods are not applicable to Problem \ref{acproblem}. As inspired by \cite{GABBP, practical_bilevel}, we decide to tackle the problem with an evolutionary heuristic approach. Evolutionary approaches are commonly applied to bilevel optimization problems in transport science. For example, in \cite{DE}, Differential Evolution (DE) is employed to address the optimal toll problem, which is about setting polls to control congestion, and the road network design problem, which determines the capacity enhancements of network facilities. In \cite{transitporiority},  GA is applied to the transit road space priority problem, which optimizes the system by reallocating the road space between private car and transit modes. We will design a GA-based algorithm to solve Problem \ref{acproblem}. Before \textcolor{black}{discussing the details of the algorithm}, we define admissibility and give some analytical results for Problem \ref{acproblem}, which can help design the algorithm in the next section.

\begin{definition}[Admissibility]
A set of requests $\check{\mathcal{R}}$ is admissible if \textcolor{black}{$[z_r]_{r\in\mathcal{R}}$ produces} $\check{\mathcal{R}}$, which results in finite profit, i.e., $\Phi(\mathcal{R},\textcolor{black}{[z_r]_{r\in\mathcal{R}}})>-\infty$.
\end{definition}
\begin{theorem} \label{thm:admissibility}
We have the following results for admissibility:
\begin{enumerate}
	\item Consider that a subset of requests $\check{\mathcal{R}}\subset \mathcal{R}$ are admissible. Let $\mathcal{P}(\check{\mathcal{R}})$ be the power set of $\check{\mathcal{R}}$. Any $\check{\mathcal{R}}'\in \mathcal{P}(\check{\mathcal{R}})$  is also admissible.
	\item For any singleton $\{r\} \subset \mathcal{R}$, if $\{r\}$ is not admissible, any superset $\check{\mathcal{R}}\supset \{r\}$ are also non-admissible.
	\item Consider subsets of requests, $\check{\mathcal{R}}_1,\check{\mathcal{R}}_2\subset \mathcal{R}$, and subsets of vehicles $\check{\mathcal{K}}_1,\check{\mathcal{K}}_2\subset\mathcal{K}$. Suppose $\check{\mathcal{K}}_1\cap \check{\mathcal{K}}_2=\emptyset$. If $\check{\mathcal{R}}_1$ and $\check{\mathcal{R}}_2$ are admissible by $\check{\mathcal{K}}_1$ and $\check{\mathcal{K}}_2$, respectively, then $\check{\mathcal{R}}_1 \cup \check{\mathcal{R}}_2$ are also admissible.
\end{enumerate}
\end{theorem}
\begin{proof}
For Statement 1, 
Constraint \eqref{admission_func} defines $\check{\mathcal{R}}$, which is an input of Constraint \eqref{bilevel_lower}. It is sufficient to show that the removal of any $r\in \check{\mathcal{R}}$ will not make \eqref{bilevel_lower} infeasible if the participating AVs can serve all the requests in $\check{\mathcal{R}}$. Suppose that $r$ is removed from $\check{\mathcal{R}}$ and AV $k$ would have assigned to serve $r$ if $r$ had been admitted. $k$ can still follow the path as if $r$ is present. Hence \eqref{bilevel_lower}  is still feasible for $\check{\mathcal{R}}\setminus r$.

For Statement 2,
a non-admissible $r$ means that it is impossible to arrange an AV to entertain $r$. We will never be able to provide \textcolor{black}{services} to a set of requests containing $r$ as its component $r$ can never be served.

For Statement 3,
we can represent $\check{\mathcal{R}}_1 \cup \check{\mathcal{R}}_2$ by three non-overlapping sets $\check{\mathcal{R}}_1\setminus(\check{\mathcal{R}}_1 \cap \check{\mathcal{R}}_2)$, $\check{\mathcal{R}}_2\setminus(\check{\mathcal{R}}_1 \cap \check{\mathcal{R}}_2)$, and $\check{\mathcal{R}}_1 \cap \check{\mathcal{R}}_2$. Since $\check{\mathcal{K}}_1$ and $\check{\mathcal{K}}_2$ are mutually exclusive, $\check{\mathcal{R}}_1\setminus(\check{\mathcal{R}}_1 \cap \check{\mathcal{R}}_2)$ and $\check{\mathcal{R}}_2\setminus(\check{\mathcal{R}}_1 \cap \check{\mathcal{R}}_2)$ can be served by $\check{\mathcal{K}}_1$ and $\check{\mathcal{K}}_2$ simultaneously. Each $r\in \check{\mathcal{R}}_1 \cap \check{\mathcal{R}}_2$ can be admitted by either $k\in\check{\mathcal{K}}_1$ or $k\in \check{\mathcal{K}}_2$.
\end{proof}

\begin{lemma} \label{nonnegativesoln}
\textcolor{black}{The system will not make negative profit. That is, for}
any $\mathcal{R}$, Problem \ref{acproblem} must have at least one feasible solution whose objective function value is non-negative.
\end{lemma}
\begin{proof}
We separate $\mathcal{R}$ into the previously admitted and newly received requests, i.e., $\{\mathcal{R}_k\}$ and $\mathcal{R}\setminus \{\mathcal{R}_k\}$. 

For the newly received requests, we can always set $z_r=0,\forall r\in\mathcal{R}\setminus \{\mathcal{R}_k\}$. Then \eqref{admission_func} gives $\check{\mathcal{R}}=\emptyset$. Eq. \eqref{bilevel_lower} returns $\alpha$ with $\phi(\alpha)=0$ as no requests need to be served and thus no AVs have been used to provide service. Hence we have $\sum_{r\in\mathcal{R}\setminus \{\mathcal{R}_k\}}\rho_r z_{r} - \phi(\alpha)=0$.

For the previously admitted requests, since they are admitted in some previous admission control processes, they must incur non-negative profit when they were admitted as new requests before. Otherwise, we would not have admitted them at the first place.
\end{proof}
Lemma \ref{nonnegativesoln} implies that Problem \ref{acproblem} must be feasible.

\begin{theorem} \label{obj_improvement}
Consider two subsets of requests $\check{\mathcal{R}}$ and $\check{\mathcal{R}}'$ with $\check{\mathcal{R}}\subset \check{\mathcal{R}}'\subset \mathcal{R}$. If  both $\check{\mathcal{R}}$ and $\check{\mathcal{R}}'$ are admissible, then \textcolor{black}{$\check{\mathcal{R}}'$ will not be less profitable than $\check{\mathcal{R}}$, i.e., } $\sup \Phi(\check{\mathcal{R}},\{z_r|r\in \check{\mathcal{R}}\})\leq \sup \Phi(\check{\mathcal{R}'},\{z_r|r\in \check{\mathcal{R}}'\})$. 
  \end{theorem}
\begin{proof}
Suppose $\sup \Phi(\check{\mathcal{R}},\{z_r|r\in \check{\mathcal{R}}\})> \sup \Phi(\check{\mathcal{R}'},\{z_r|r\in \check{\mathcal{R}}'\})$. We write $\check{\mathcal{R}}' = \check{\mathcal{R}}\cup (\check{\mathcal{R}}'\setminus \check{\mathcal{R}})$. Then we have 
\begin{align*}
&\sup \Phi(\check{\mathcal{R}'},\{z_r|r\in \check{\mathcal{R}'}\}) \\
&= \sup \Phi(\check{\mathcal{R}},\{z_r|r\in \check{\mathcal{R}}\})+ \sup \Phi(\check{\mathcal{R}'}\setminus \check{\mathcal{R}},\{z_r|r\in \check{\mathcal{R}'}\setminus \check{\mathcal{R}}\}).
\end{align*} 
By Lemma \ref{nonnegativesoln},  $\sup \Phi(\check{\mathcal{R}'}\setminus \check{\mathcal{R}},\{z_r|r\in \check{\mathcal{R}'}\setminus \check{\mathcal{R}}\})$ has a value larger than or equal to zero. This induces a contradiction.
\end{proof}
Theorem \ref{obj_improvement} implies that entertaining more requests will not reduce the amount of profit made.

\subsection{\textcolor{black}{Variations}}

\textcolor{black}{Here we investigate how traffic congestion and no-show of paasengers impact on admission control (and scheduling). Basically, we will see that under these circumstances, the proposed admission control and scheduling mechansims can still be applied but we may need some additional minor arrangements to handle various situations.}

\subsubsection{\textcolor{black}{Traffic Congestion}}
\textcolor{black}{Traffic congestion has direct impact on the travel time $t_{ij}$ for some $(i,j)\in\mathcal{E}$ and subsequently affects the admissibility of requests. Recall that the system operates in a fixed-interval basis and each interval generally lasts for a few minutes (see Section \ref{subsec:op}). We basically assume that, within an interval, the parameters, including the travel times, are constant or with very small changes such that the results of admission control completed for that interval are still valid. If the travel times are relatively fast changing, we need to shorten the duration of the intervals to make the assumption valid. On the other hand, if the travel times are slowly varying, we may lengthen the durations to reduce the computation burden. Hence, the duration of the operating intervals depends on the traffic conditions of the deployed service area.}

\textcolor{black}{Now consider that $t_{ij}$ in the current interval has been updated such that its value is different from that used in the previous interval. There are three cases for the possible influence: (i) $t_{ij}$ does not involve in $\mathcal{R}_k$ for all $k$; (ii) $t_{ij}$ involves in $\overline{\mathcal{R}}_k$ for some $k$; and (iii) $t_{ij}$ involves in $\tilde{\mathcal{R}}_k$ for some $k$. For Case (i), since $t_{ij}$ has not been used to serve any requests, its change does not affect the schedules of any AVs. Hence, nothing needs to be done solely based on $t_{ij}$. For Case (ii), although $t_{ij}$ has been used to determine the schedules of some AVs, the involved requests have not been served yet. We can simply consider these requests as newly submitted requests and perform admission control and scheduling with them again. For Case (iii), $t_{ij}$ affects those schedules which are being implemented by some AVs. In the subsequent intervals, the scheduling process will see if the road segment $(i,j)$ can be avoided by determining other shortest paths. If not, as the passengers are being served, it may not be appropriate to ask them to shift to other vehicles for their journeys and nothing can be done further operationally. However, we may compensate the passengers in the marketing perspective, e.g., by issuing cash coupons for future rides.}

\subsubsection{\textcolor{black}{No-show of Passengers}}
\textcolor{black}{No-show refers to the situation that some or all passengers of a paricular request are absent at the scheduled pickup time. If a passenger cannot arrive at the pickup location on time, this will be considered as no-show.  If some but not all passengers are absent, the schdule of the designated AV is unaffected but fewer seats are required. These unused seats can be released to serve other appropriate requests in the later intervals.  If all passengers are absent, the ``resources'' allocated to the request can be released in the subsequent intervals right after its original pickup time. This gives the AV more flexibility in time and occupancy to serve future requests. In the business perspective, there may exist some penalty policies to discourage such activities. }

\section{Genetic-Algorithm-Based Solution Method}\label{sec:algorithm}
In this section, we propose a solution method to tackle Problem \ref{acproblem}. We adopt a GA-based framework to structure the method. Some of its components are designed based on the analytical results discussed in Section \ref{sec:admission}.

\subsection{Working Principle of Evolutionary Algorithms}

Evolutionary Algorithms (EAs) refer to a class of optimization algorithms, whose designs are inspired by various natural phenomena. Examples include GA \cite{GA}, DE \cite{DE_book}, and Chemical Reaction Optimization (CRO) \cite{CRO}. Different EAs generally have similar working principles: An EA samples the solution space of the problem iteratively and tries to locate \textcolor{black}{a} global optimum after examining a limited \textcolor{black}{number} of candidate solutions in the solution space. 
In each iteration, with some operators, it generates a population of candidate solutions based on those obtained from the previous iterations and their corresponding objective function values. It tends to converge to the global optimum along the iterations and it terminates when a stopping criterion is matched. Different EAs have different designs of their  operators. For example, GA is designed based on the ideas of natural selection in \textcolor{black}{genetics} while CRO mimics the nature of chemical reaction processes. Unlike most of the traditional optimization approaches, EAs \textcolor{black}{require} the problem to be neither convex nor differentiable. In each algorithm run, they only need to sample a number of candidate solutions  and evaluate their solution qualities with the objective function. Hence, a search with an EA usually incurs many objective function calls.
\textcolor{black}{As discussed, EAs have been shown effective in solving bilevel optimization problems in transport science. We are going to adopt the well-established GA framework to facilitate the design of a method which can return good solutions for Problem \ref{acproblem} in a practical sense.}

\subsection{Distributed Scheduling} \label{subsec:distributed}

When an EA is employed to address Problem \ref{acproblem}, many candidate solutions will be generated. To evaluate the quality of a particular candidate solution, we need to compute \eqref{bilevel_obj} once, which also needs to examine \eqref{bilevel_lower} one time. In other words, a single run of EA requires to solve Problem \ref{schedulingproblem} many times. When the lower-level optimization is simple, the computational burden of solving it many times may still be acceptable. However, this is not the case for Problem \ref{schedulingproblem}, where the required \textcolor{black}{numbers of} variables and constraints grow exponentially with the \textcolor{black}{quantities} of transportation requests and serving AVs. This implies that we need to a more effective way to solve Problem \ref{schedulingproblem}, in order to tackle Problem \ref{acproblem}.

\textcolor{black}{Consider that} $\check{\mathcal{R}}_k\subset \mathcal{R}$ \textcolor{black}{is} the subset of requests assigned to vehicle $k$. 
Suppose that we know the distribution of the requests to the vehicles, i.e., $\check{\mathcal{R}}_k$ \textcolor{black}{for all $k$}.
Since each request is only served by one vehicle, we have $\check{\mathcal{R}}_k \cap \check{\mathcal{R}}_l = \emptyset$, for any $k,l\in \mathcal{K}$\textcolor{black}{, $k\neq l$}, and $\bigcup_{k\in\mathcal{K}}\check{\mathcal{R}_k} = \mathcal{R}$.
When given $\check{\mathcal{R}}_k$, we consider the following problem:
\begin{problem}[Scheduling Subproblem for vehicle $k$] \label{subproblem}
\begin{subequations}
\label{admissioncontrol}
\begin{align}
\text{maximize}\quad 	& \sum_{i,j\in \mathcal{V}'} c_{ij}\mathring{x}_{ij}^k\\
\text{subject to}\quad 
& \sum_{i\in\tilde{\mathcal{V}}}\mathring{g}_i^{k} \leq 1, \\ 
& 0\leq \sum_{i\in \mathcal{N}^-(a_k)}\mathring{x}_{a_ki}^k - \sum_{i\in \mathcal{N}^+(a_k)}\mathring{x}_{ia_k}^k \leq \sum_r{\mathring{y}_r^k},\\
& 0\leq \sum_{j\in \mathcal{N}^+(i)}\mathring{x}_{ji}^k - \sum_{j\in \mathcal{N}^-(i)}\mathring{x}_{ij}^k \leq \mathring{g}_i^k,\forall i\in \tilde{\mathcal{V}}, \\
& \sum_{j\in \mathcal{N}^+(i)}\mathring{x}_{ji}^k = \sum_{j\in \mathcal{N}^-(i)}\mathring{x}_{ij}^k, \forall i\in\mathcal{V}'\setminus \tilde{\mathcal{V}}\cup\{a_k\} \\
& \sum_{i\in\mathcal{N}^-(s_r)}\mathring{x}_{s_ri}^k \geq \mathring{y}_r^k, \forall r\in\check{\mathcal{R}}_k,\\
& \sum_{i\in\mathcal{N}^+(d_r)}\mathring{x}_{id_r}^k \geq \mathring{y}_r^k, \forall r\in\check{\mathcal{R}}_k,\\
& \mathring{t}_k^0\leq \mathring{t}_i^k \leq \tilde{T}_k, \forall i\in \mathcal{V}',\\
& \mathring{t}_j^k \geq \mathring{t}_i^k + \mathring{t}_{ij} - M(1-\mathring{x}_{ij}^k), \forall i,j\in\mathcal{V}'\\
& \mathring{t}_{d_r}^k - \mathring{t}_{s_r}^k \leq T_r + M(1-\mathring{y}_r^k), \forall r\in \check{\mathcal{R}}_k,\\
& e_r-M(1-\mathring{y}_r^k)\leq \mathring{t}_{s_r}^k\leq l_r+M(1-\mathring{y}_r^k),\forall r\in \check{\mathcal{R}}_k,\\
& 0\leq \mathring{f}_i^k \leq Q_k, \forall i\in\mathcal{V}',\\
& \mathring{f}_{a_k}^k \geq \sum_{r|s_r=a_k} q_r\mathring{y}_r^k - \sum_{r|d_r=a_k} q_r\mathring{y}_r^k,\\
& \mathring{f}_{j}^k \geq \mathring{f}_{i}^k - M(1-\mathring{x}_{ij}^k)+ \sum_{r|s_r=a_k, r\in \check{\mathcal{R}}_k} q_r\mathring{y}_r^k \nonumber\\
&\quad - \sum_{r|d_r=a_k, r\in \check{\mathcal{R}}_k} q_r\mathring{y}_r^k, \forall i,j\in\mathcal{V}',\\
& \mathring{f}_i^k\leq M(1-\mathring{g}_i^k), \forall i\in \tilde{V},\\
\text{over}\quad & \mathring{x}_{ij}^k\in\{0,1\}, \mathring{y}_r^k\in\{0,1\}, \mathring{g}_l^k\in\{0,1\}, \mathring{t}_i^k\in\mathbb{R}^+, \nonumber\\
& \mathring{f}_i^k\in\mathbb{Z}^+,  \forall i,j\in \mathcal{V}', l\in \tilde{\mathcal{V}}, r\in\check{\mathcal{R}}_k. 
\end{align}
\end{subequations}
\end{problem}
Solving Problem \ref{subproblem} only allows us to obtain the serving path, the schedule to reach various locations along the path, and the capacity conditions of vehicle $k$ for serving the requests indicated by $\check{\mathcal{R}}_k$. Problem \ref{subproblem} looks similar to Problem \ref{schedulingproblem} but indeed much simpler. It does not contain \eqref{requestassignment} and it manipulates fewer variables as those related to vehicles other than $k$ are not included. It also possesses fewer constraints because of fewer variables.

For simplicity, similar to \eqref{sim_scheduling}, we also write the solution, objective function, and the solution space of Problem \ref{subproblem} as $\mathring{\alpha}_k$, $\phi_k(\mathring{\alpha}_k)$ and $\mathcal{Z}_k$\textcolor{black}{, respectively.}

\begin{theorem} \label{distributedScheduling}
When given $\check{\mathcal{R}}_k\subset \mathcal{R}, \forall k\in\mathcal{K}$, such that $\check{\mathcal{R}}_k \cap \check{\mathcal{R}}_l = \emptyset$, for any $k,l\in \mathcal{K}$\textcolor{black}{, $k\neq l$}, and $\bigcup_{k\in\mathcal{K}}\check{\mathcal{R}}_k = \mathcal{R}$, solving Problem \ref{subproblem} for all $k\in\mathcal{K}$ is equivalent to solving Problem \ref{schedulingproblem}, i.e.,
\begin{align*}
	\inf_{\alpha\in\mathcal{Z}} \phi(\alpha) = \sum_{k\in\mathcal{K}} \inf_{\mathring{\alpha}_k\in\mathcal{Z}_k} \phi_k(\mathring{\alpha}_k),
\end{align*}
and $x_{ij}^k =\mathring{x}_{ij}^k$, $y_{r}^k =\mathring{y}_{r}^k$, $g_{l}^k =\mathring{g}_{l}^k$, $t_{i}^k =\mathring{t}_{i}^k$, and $f_{i}^k =\mathring{f}_{i}^k$, $\forall i,j\in \mathcal{V}', l\in \tilde{\mathcal{V}}, r\in\check{\mathcal{R}}_k, k\in\mathcal{K}$.
\end{theorem}

\begin{proof}
When given such $\check{\mathcal{R}}_k\subset \mathcal{R}, \forall k\in\mathcal{K}$, we can construct $y_r^k, \forall r\in\mathcal{R},k\in\mathcal{K}$, such that  \eqref{requestassignment} holds. In this way, we can remove \eqref{requestassignment} from Problem \ref{schedulingproblem}. Without Constraint \eqref{requestassignment}, the objective function and the rest of the constraints of Problem \ref{schedulingproblem} become separable in terms of $k$: \eqref{objective} gives the sum of costs spent on the vehicles; Eqs. \eqref{gasAssignment}--\eqref{posdesflow} specify the paths traversed by the vehicles, each of which are independent; Eqs. \eqref{timescope}--\eqref{window} confine the time requirements at various locations along the vehicular paths; Eqs. \eqref{capCon}--\eqref{refuelCap} limit the passenger capacity conditions along the vehicular paths. If we group the terms of \eqref{objective} and the constraints \eqref{gasAssignment}--\eqref{refuelCap} for each $k$, we will have $|\mathcal{K}|$ problems, each of which is given by \eqref{admissioncontrol}.
\end{proof}
\textcolor{black}{Theorem \ref{distributedScheduling} states that when the assignment of requests to the vehicles is known, solving the $|\mathcal{K}|$ individual scheduling subproblems distributedly can retain the solution of the original scheduling problem. Note that this result is dedicatedly developed based on some characteristics of the problem formulations and it generally cannot be applied to the other scheduling problems. Unlike general distributed optimization \cite{para_dist,multiagent}, our result here does not require techniques like message-passing. As a result, the $|\mathcal{K}|$ subproblems can be solved by $|\mathcal{K}|$ computing units distributedly. Assuming that the vehicles are connected through advanced vehicular communication technologies at all times, an obvious option of the computing unit is the AV. Thus,}
by Theorem \ref{distributedScheduling}, if we can assign each vehicle with the requests it needs to serve, each vehicle can determine a feasible  path \textit{per se} to serve the assigned requests with the lowest cost by solving Problem \ref{subproblem} concurrently.
\textcolor{black}{However, when the communications between a particular AV and the control center are interrupted, the corresponding subproblem can be delegated to an unoccupied computing unit at the control center or even to the cloud instead. The computed scheduling result can be returned to the AV when its communications have been resumed.}



\subsection{Algorithmic Components}

Since GA is one of the most popular EAs, we adopt a GA-based design to address the admission control problem. GA generates a sequence of candidate solutions using operations inspired by natural evolution, e.g., inheritance, selection, crossover, and mutation. Here we introduce various algorithmic components before discussing the overall algorithmic design:

\subsubsection{Chromosome}
A chromosome specifies  a candidate solution of Problem \ref{acproblem}. While the lower-level optimization is handled by a standard MILP method, our GA approach is mainly used to handle the upper-level optimization. A chromosome is represented by a $1\times |\mathcal{R}|$ binary vector $z=[z_1,\ldots,z_r,\ldots,z_{|\mathcal{R}|}]$,  together with a vehicle assignment vector $\kappa=[\kappa_1,\ldots,\kappa_r,\ldots,\kappa_{|\mathcal{R}|}]$, where $\kappa_r$ represents the vehicle assigned to $r$ if $z_r$ is of unity. Note that the introduction of $\kappa$ is the trick to carry out distributed scheduling discussed in Section \ref{subsec:distributed}. Although $\kappa$ can be determined in \eqref{bilevel_lower} if we apply the original formulation of scheduling \eqref{sim_scheduling}, there is no harm in manipulating $\kappa$ together with $z$ in the chromosome level. This makes distributed scheduling feasible and the benefit of computation time saving will be clear in Section \ref{subsec:scheduling}.  
During the course of search, we maintain a population of $N_{pop}$ chromosomes.

\subsubsection{Fitness Evaluation}
We evaluate the fitness of each chromosome in a distributed manner. The fitness evaluation process is illustrated in Fig. \ref{fig:fitness} and it consists of five steps:
\begin{enumerate}
\item[(1)] Grouping requests in $\check{\mathcal{R}}_k$:
Each chromosome $i$ contains $z^i$ and $\kappa^i$. For those $r$'s with $z^i_r=1$, based on $\kappa^i$, at the control center, we can divide $\mathcal{R}$ into $|\mathcal{K}|$ groups, i.e., $\check{\mathcal{R}}_k, \forall k\in\mathcal{K}$.
\item[(2)] Request information distribution:
For each $k$, the control center transmits $\check{\mathcal{R}}_k$ to AV $k$\textcolor{black}{, e.g.,} via VANET.
\item[(3)] Distributed scheduling:
Modern vehicles are generally equipped with computers and thus each $k$ can solve the individual Problem \ref{subproblem} simultaneously with other vehicles. Those AVs with empty $\check{\mathcal{R}}_k$ assigned can skip the computation.
\item[(4)] Individual cost return:
The individual vehicles transmit the computed costs of scheduling to the control center\textcolor{black}{, e.g.,} via VANET.
\item[(5)] Fitness computation:
Based on Theorem \ref{distributedScheduling}, the cost associated to the chromosome is the sum of the objective function values of Problem \ref{subproblem} determined by the individual AVs, i.e., 
$\phi(\alpha) = \sum_{k\in\mathcal{K}}\phi_k(\mathring{\alpha}_k|\kappa_r=k)$. Then the fitness of the chromosome can be computed as $\Phi(z,\kappa) = \sum_{r\in\mathcal{R}}\rho_rz_r - \sum_{k\in\mathcal{K}}\phi_k(\mathring{\alpha}_k|\kappa_r=k)$.\footnote{By abuse of notation, we write $\Phi(z,\kappa)=\Phi(\mathcal{R},\textcolor{black}{[z_r]})$ to emphasize the structure of the chromosome.}
\end{enumerate}
Note that $\Phi(z,\kappa)$ becomes $-\infty$ if and only if any request $r$ with $z_r=1$ is non-admissible. The advantages of undergoing the above process are three-fold: 
\begin{enumerate}
\item[(i)] The computation time can be dramatically reduced. Among all the computation components in the algorithm, scheduling is the most computationally demanding. If each vehicle can compute their own schedules, all the individual scheduling subproblems can be solved simultaneously.
\item[(ii)] All entities need to manage the necessary data only. $\rho_r$ is the result of the deal between the customer and the control center. \textcolor{black}{With distributed scheduling, the} usage of $\rho_r$ is restricted to the control center and no vehicles are involved. Moreover, after a vehicle solves its scheduling subproblem, its computed schedule is stored in that vehicle only, but not the control center nor any other vehicles. 
\item[(iii)] The amount of communications keeps minimal. In each evaluation, the only data needed to be communicated between the control center and the vehicles are the requests assigned to the individual vehicles \textcolor{black}{(in Step 2)} and the computed scheduling costs \textcolor{black}{(in Step 4)}. The system does not require a sophisticated communication system to satisfy the communication requirements.
\end{enumerate}

\begin{figure}[!t]
\centering
\hspace{-0.3cm}
\includegraphics[width=2.7in]{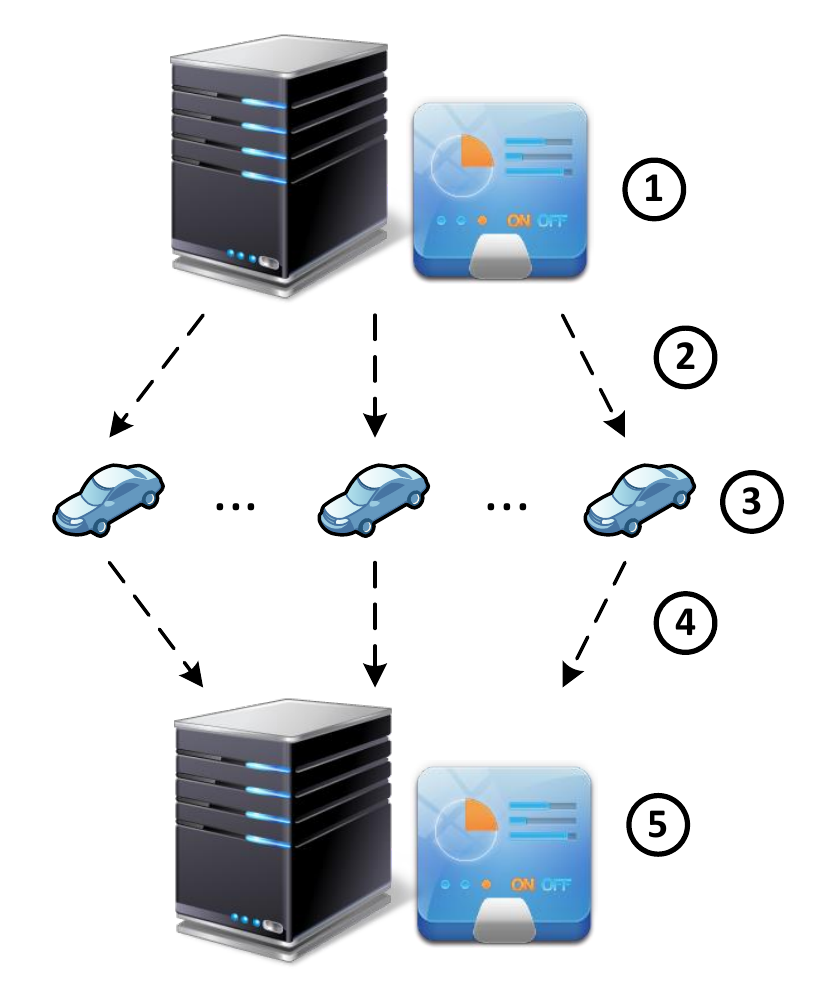} \vspace{-0.5cm}
\caption{Fitness evaluation process.} 
\label{fig:fitness}
\end{figure}

\subsubsection{Tabu List}

We construct a tabu list $\tau_r$ for each request $r$ to reduce the size of the search space. $\tau_r$ contains those vehicles $k$ which cannot serve $r$. As implied by Theorem \ref{thm:admissibility}, if a request $r$ is not admissible by $k$, any set of requests containing $r$ will also \textcolor{black}{not be} admissible by $k$. In other words, we will never need to consider those $k$ in $\tau_r$ when configuring $\kappa_r$. Unlike Tabu Search \cite{tabu}, we do not need to update the tabu lists during the course of search.\footnote{\textcolor{black}{As discussed in Section \ref{subsec:op}, admission control is completed in the duty assignment sub-interval once in each operating interval. Such sub-interval is short so that it is unlikely to have great changes to the positions of the AVs. Thus the tabu lists can be assumed to be static throughout the admission control process happened in each interval. However, the tabu lists may need to be updated in the next interval as the vehicles may have moved to other positions.}} $\tau_r$'s are only constructed in the initialization phase of the algorithm and utilized in \textcolor{black}{both} initial population generation and mutation.

\subsubsection{Selection}
In each generation, a fraction $X_{rate}$ of $N_{pop}$ survives and the rest of $(1-X_{rate})$ will be replaced by the children bled in the processes of crossover. We apply weighted random pairing \cite{PGA} to select the survived chromosomes to perform crossover.

\subsubsection{Crossover}
Crossover is an operator in GA to achieve intensification. In each operation, it manipulates two parent chromosomes to breed two offspring. The offspring inherit the merits from their parents and thus they tend to have better fitness values, i.e., higher objective function values of \eqref{bilevel_obj}. By Theorem \ref{obj_improvement}, a larger set of requests will improve the fitness. Also based on Statement 3 of Theorem \ref{thm:admissibility}, we manipulate the chromosomes with crossover as follows.
Parents $i$ and $j$ reproduce offspring $i'$ and $j'$. $i'$ admits all those $r$'s as $i$ does with the same set of vehicles. If there is any $k$ which is adopted in $j$ but not in $i$, we randomly adopt one such $k$ in $i$ on those $r$'s which are not admitted in its parent $i$.  We produce an offspring $j'$ dominantly inherited by the parent $j$ similarly. In this way, the offspring are likely to admit more requests resulting in higher fitness.

\subsubsection{Mutation}
Mutation exhibits diversification to prevent the algorithm from getting stuck in local optimums and we basically follow \cite{PGA} to design mutation. We control the amount of mutation with a mutation rate $\mu\in[0,1]$. 
We apply elitism to the chromosome with highest fitness in the population and only the rest undergo mutation. A mutation occurs on bit $z_r^i$ of chromosome $i$ and the number of mutations taken place in each generation is $\mu\times(N_{pop}-1)\times |\mathcal{R}|$. 
If we perform mutation on $z_r^i$, we toggle  $z_r^i$. If $z_r^i$ is changed from 0 to 1, we randomly assign $\kappa_r$ a $k$ which is not in the tabu list $\tau_r$. If $z_r^i$ is changed from 1 to 0, we set $\kappa_r=0$.
To further enhance diversification, besides the elite chromosome, each chromosome has a probability of $\gamma$ to be replaced by a random chromosome.

\subsection{Algorithmic Design}
We basically follow \cite{PGA} to design the algorithm, which consists of three stages: initialization, iterations, and the final stage. The flow chart of the algorithm is given in Fig. \ref{fig:flowchart}. We maintain the chromosomes with feasible candidate solutions during the whole course of search.

\subsubsection{Initialization}
In initiation, we define all the system parameters, e.g., $N_{pop}$ and $X_{rate}$, and \textcolor{black}{construct the tabu list $\tau_r$ for each $r$. Then we} create the initial population of chromosomes\textcolor{black}{, each of which} is assigned with one random request $r$ associated with a vehicle not in its tabu list $\tau_r$. This can ensure all chromosomes are initially feasible. We evaluate the fitness of the initial chromosomes before the iterations start. 
\subsubsection{Iterations}
In each iteration (or called generation), we manipulate the candidate solutions held by the chromosomes. Before any modification, we back up the feasible candidate solutions stemmed from the previous generation. Then we perform selection, crossover, and mutation to manipulate the chromosomes, followed by fitness evaluations. If any chromosome possesses an infeasible solution, we retain its original feasible one from the backup. We check the stopping criteria to see if we continue with the next iteration or proceed to the final stage. One commonly used stopping criterion is termination after undergoing a certain number of generations.
\subsubsection{Final Stage}
We output the best solution found in this stage.

\begin{figure}[!t]
\centering
\hspace{-0.3cm}
\includegraphics[width=3.5in]{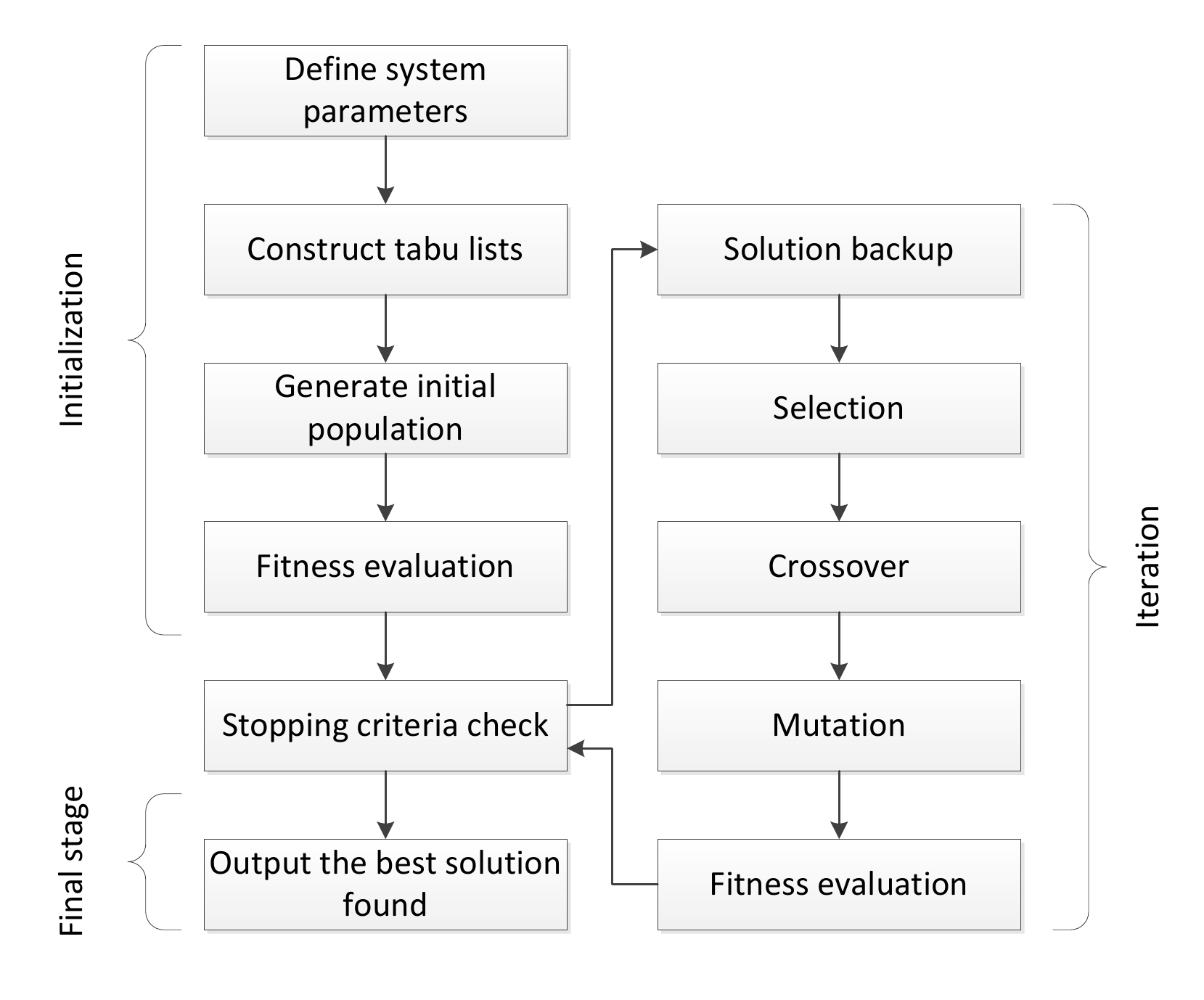} \vspace{-0.5cm}
\caption{Flow chart of the algorithm.} 
\label{fig:flowchart}
\end{figure}

In general, the solution method is implemented in a central manner at the control center. When evaluating the fitness of the chromosomes, the scheduling tasks are distributed to the vehicles \textcolor{black}{based on distributed scheduling}.

\section{Performance Evaluation}\label{sec:simulation}

We perform a series of simulations to evaluate different aspects of the algorithm. We \textcolor{black}{consider} a set of real taxi \textcolor{black}{service} data from \cite{boston}, containing the pickup and dropoff times, and pickup and dropoff locations of a number of taxi trips served in the City of Boston. We sample 100 trip data whose pickup times happened within a period of 30 minutes in a day of 2012 as the transportation request pool. \textcolor{black}{Since no existing transport can offer flexible shared-ride services as our system does, we} adopt the data for our system as follows: the earliest service starting time as the pickup time of the data, the latest service starting time as the pickup time plus 15 minutes, the maximum ride time as the actual trip time times 1.5, random seat occupancy in the range of $[1,5]$, and 50\% of the actual taxi fare \textcolor{black}{as the charges}. The driving distance and travel time between any two locations are determined through the Google Maps API. Based on \cite{fuelcost}, we assume that the fuel cost is 16 cents per mile. We select five gas stations in Boston as the refuel stations for AVs. Each vehicle is assumed to be equipped with five seats and we randomly place the vehicles in the city.

We perform the simulations on a computer with Intel Core i7-2600 CPU at 3.40 GHz and 32 GB of RAM. They are conducted in the MATLAB environment, where the scheduling problem is addressed with YALMIP \cite{yalmip} and CPLEX \cite{cplex}. We follow \cite{PGA} to set the GA parameters: $N_{pop}=16$, $X_{rate}=0.5$, and $\mu=0.15$, and we set $\gamma=0.5$.
Recall that, to operate the system for a period of time, we need to do admission control for each operating interval \textcolor{black}{within} the period. To perform admission control for an interval, we need to undergo a number of scheduling processes.
We try to evaluate the performance of the algorithm incrementally from the smallest module. First we evaluate the computation time for scheduling. In the second test, we evaluate the performance of the algorithm on solving the admission control problem. At last, we examine the profits made when the system operates continuously for a period of time.

\subsection{Computation Time for Scheduling} \label{subsec:scheduling}

\begin{figure}[!t]
\centering
\hspace{-0.3cm}
\includegraphics[width=3.5in]{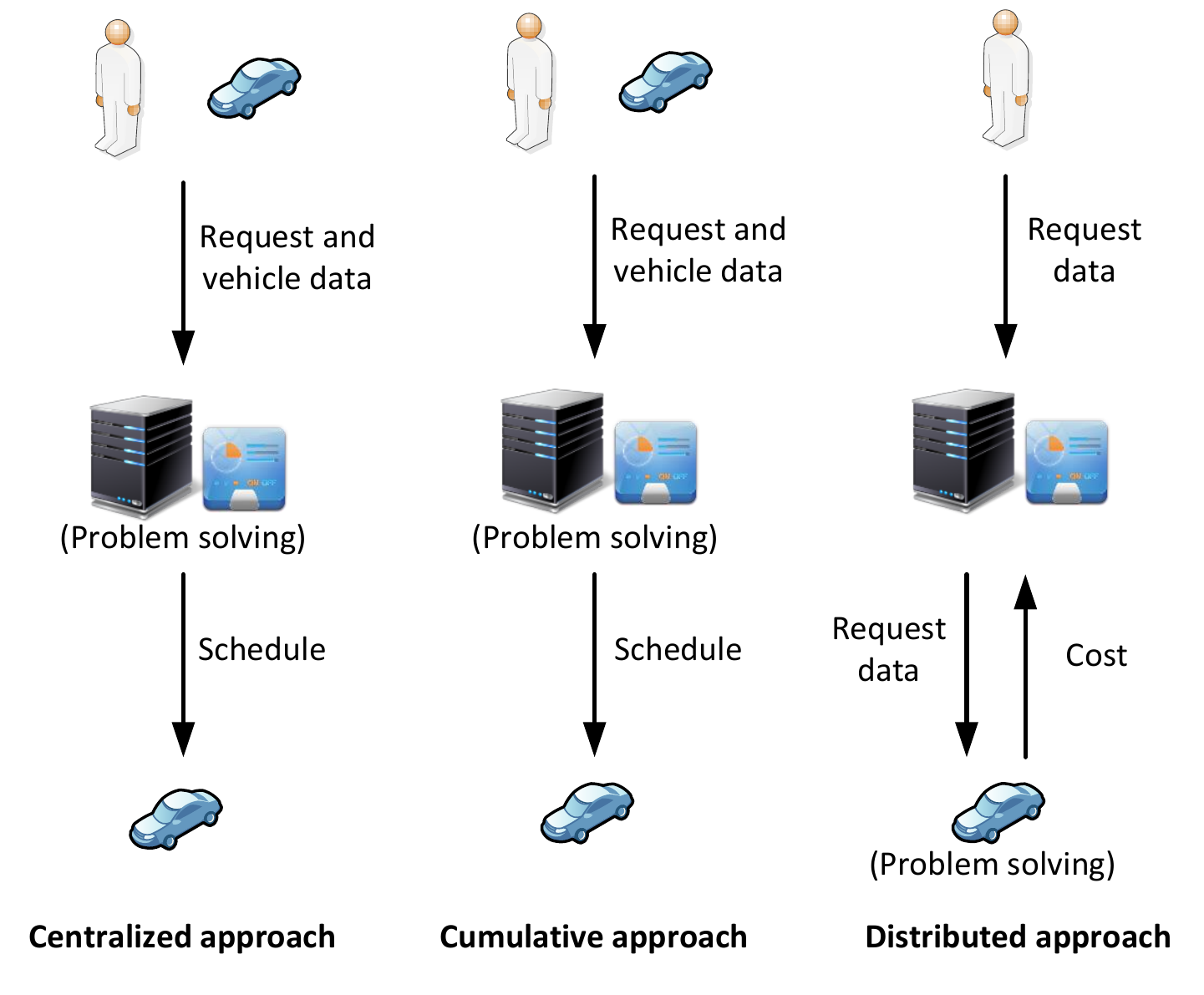} \vspace{-0.5cm}
\caption{Data processing, communications, and computation of the three approaches in scheduling.} 
\label{fig:comm}
\end{figure}

\begin{figure}[!t]
	\begin{center}
		\subfigure[3 requests]{\label{fig:time1}\includegraphics[width=3.5in]{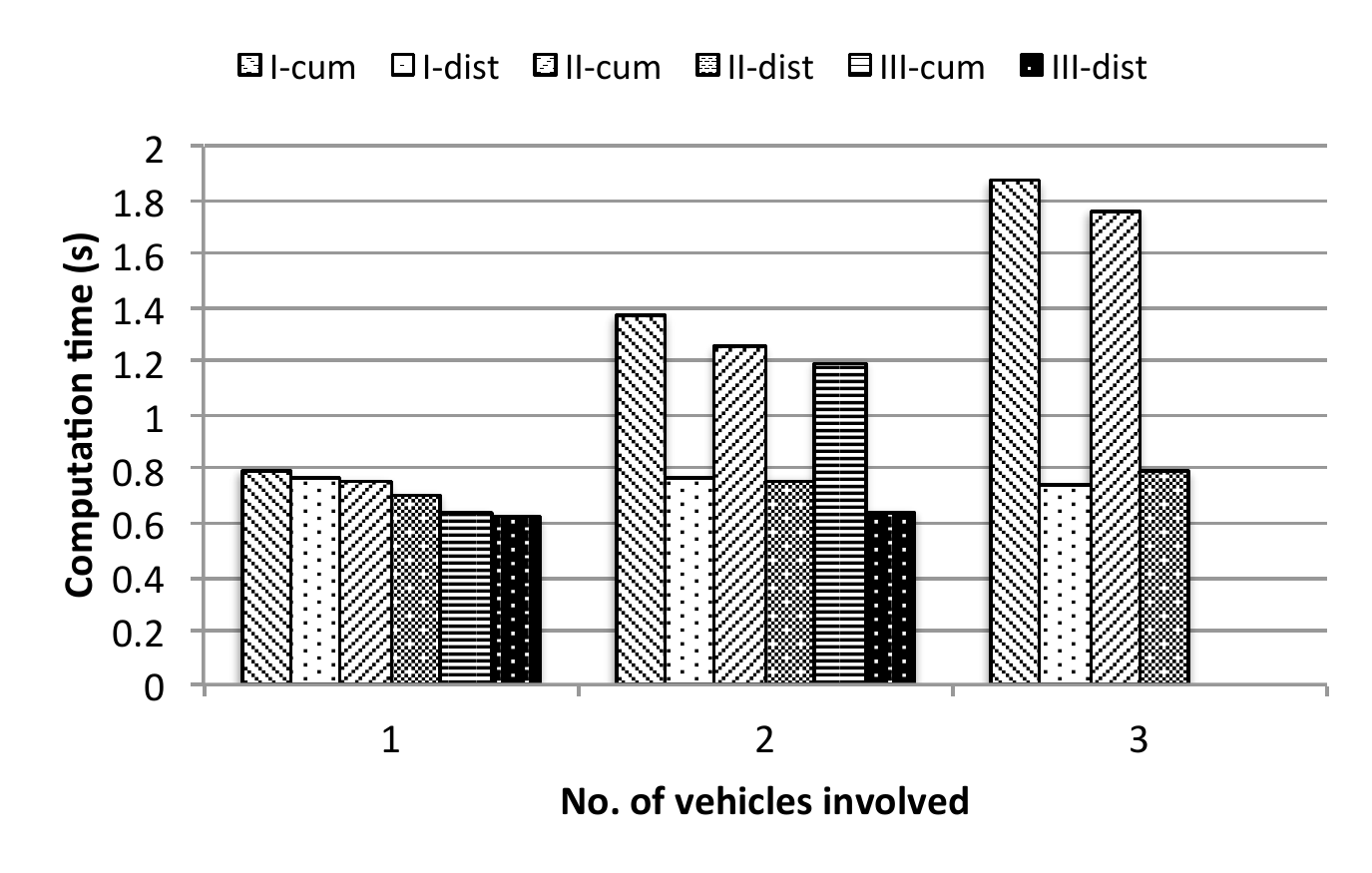}}
    	\subfigure[4 requests]{\label{fig:time2}\includegraphics[width=3.5in]{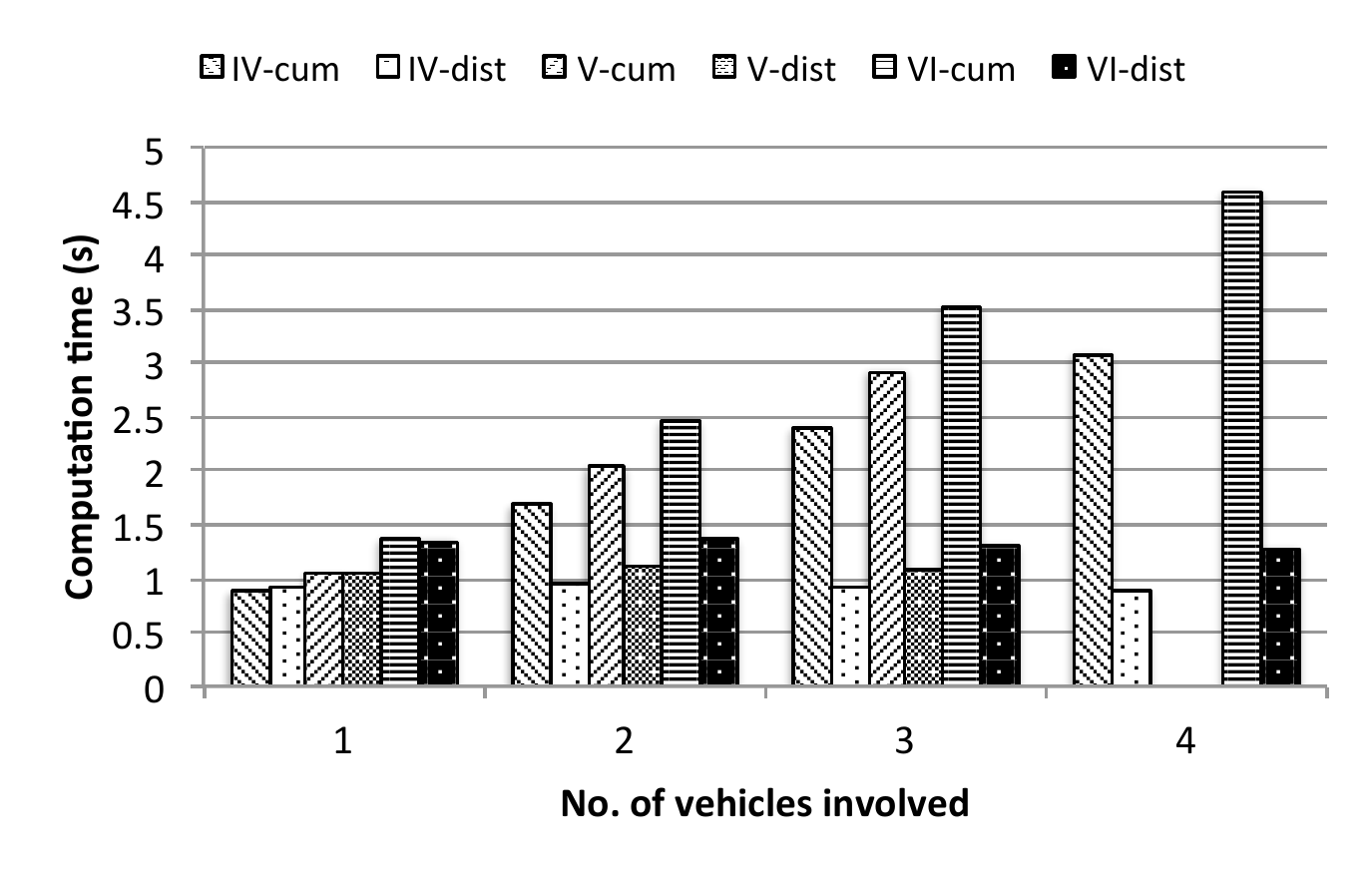}} 
    	\subfigure[5 requests]{\label{fig:time3}\includegraphics[width=3.5in]{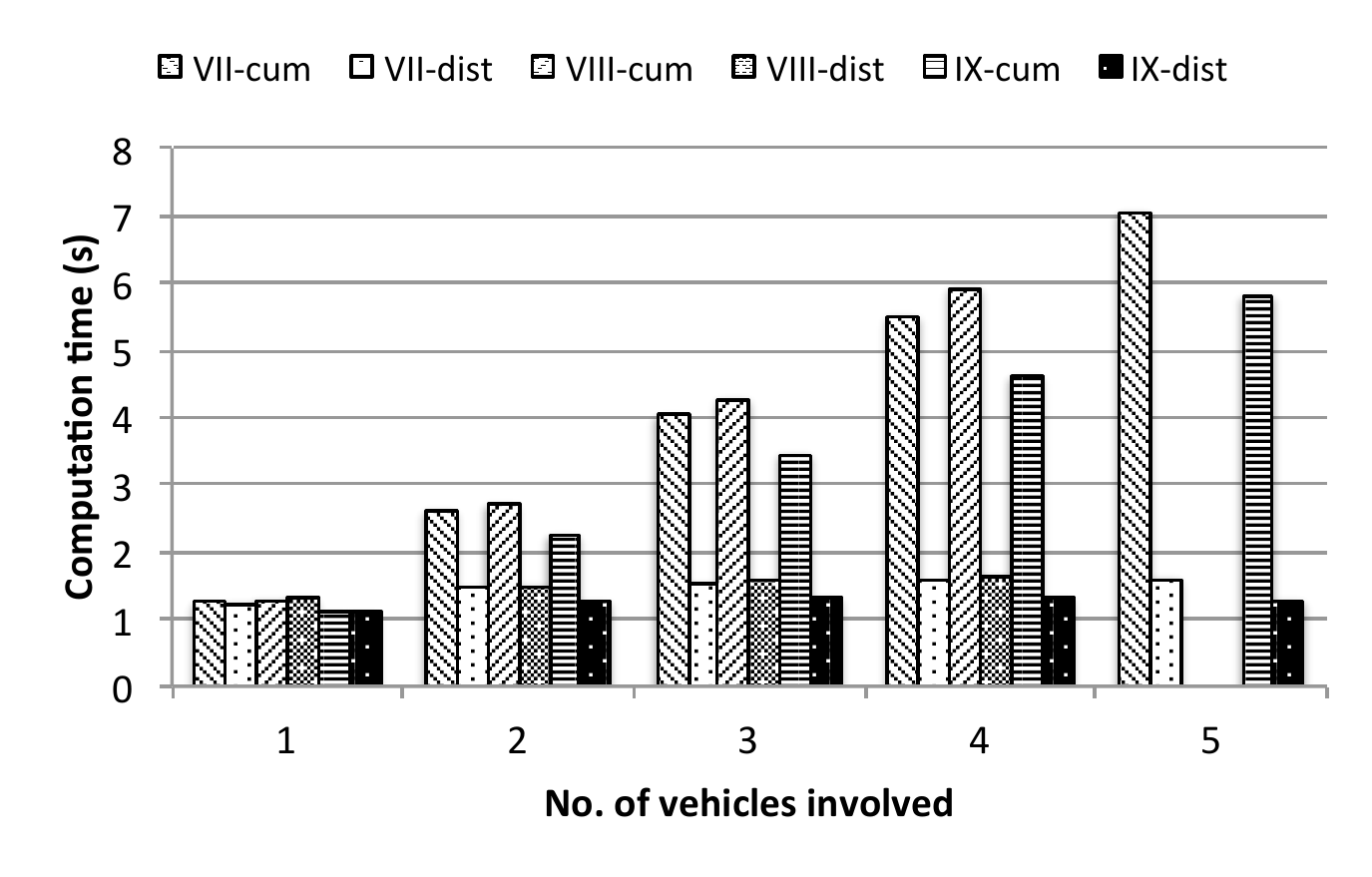}}
	\end{center}
	\caption{Computation times for scheduling.}
  \label{fig:comp_time}
\end{figure}

As Problem 1 is \textcolor{black}{an} MILP, we assume that CPLEX can return the optimal solution if the problem is tractable. So we focus on the computation time. 
When we look at Problem 1, the numbers of variables and constraints grow exponentially with the problem size in terms of the \textcolor{black}{quantities} of transportation requests and vehicles. Hence the computation time for scheduling grows very fast with the problem size. For demonstrative purposes, we focus on small problem instances. We randomly generate 9 cases from the Boston dataset: three cases with three requests, three with four requests, and three with five requests. All the cases are served with five vehicles.
Recall that we have two main ways to address the scheduling problem: (1) by solving Problem 1 as a whole and (2) by solving a number of Problem 3 collectively. For the latter, we can further arrange the subproblems \textcolor{black}{to be} solved (2.1) \textit{en masse} at the control center or (2.2) separately at the individual vehicles. Thus, there are three approaches in total and we call (1), (2.1), and (2.2) the centralized, cumulative, and distributed approaches, respectively. 
The data processing, communications, and computation of the three approaches are depicted in Fig. \ref{fig:comm}.
For the centralized and cumulative approaches, all data need to be collected and gathered at the control center from the passengers and vehicles for processing. After scheduling, the computed schedules will be \textcolor{black}{distributed} to the corresponding vehicles. For the distributed approach, the \textcolor{black}{vehicular} data are only maintained at the problem solving agents, i.e., that vehicles \textit{per se}, before and after the corresponding subproblems being solved. After scheduling, the resulting costs are transmitted back to the control center for the subsequent scheduling. When different numbers of vehicles are involved, the computation time can be noticeably different. To see this, for each of Cases I-IX, \textcolor{black}{we} examine all possible combinations of $z$ and $\kappa$ (i.e., candidate solutions for chromosomes) and check their computation times for scheduling. We consider the time spent on communications negligible as it is usually much smaller when compared with the computation time. Fig. \ref{fig:comp_time} shows the average computation times for feasible schedules with different numbers of vehicles involved in each case. Since the computation time \textcolor{black}{of the centralized approach} grows too fast (e.g.,  8.30 s, 69.25 s, and $6.72\times 10^3$ s for 3--5 requests, respectively), the time changes for the cumulative and distributed approaches would have become indistinguishable if the centralized data had also been \textcolor{black}{displayed}. For clearer representation, we skip the results for the centralized approach \textcolor{black}{in Fig. \ref{fig:comp_time}}. In Fig. \ref{fig:comp_time}, some bars are missing because no feasible schedule can be computed with particular numbers of vehicles involved. For example, one request in Case III cannot be scheduled with any vehicle, and thus, no results are shown for three vehicles for Case III. Generally, for the cumulative approach, the computation time grows linearly with the number of vehicles involved as more subproblems with similar size need to be solved. For the distributed approach, the computation times with different vehicle sizes are more or less similar because the involved subproblems can be handled at different vehicles simultaneously. While the computation time of the centralized approach grows exponentially with the number of requests, that of the cumulative approach increases at a much slower rate and that of the distributed approach is approximately steady. Hence, it is not feasible to adopt the centralized approach. If the vehicles have sufficient communication and computation capabilities, we suggest the distributed approach. Otherwise, we \textcolor{black}{can only} endorse the cumulative approach for scheduling.

\subsection{Admission Control in an Operating Interval}

\begin{figure*}[!t]
\centering
\hspace{-0.3cm}
\includegraphics[width=7.2in]{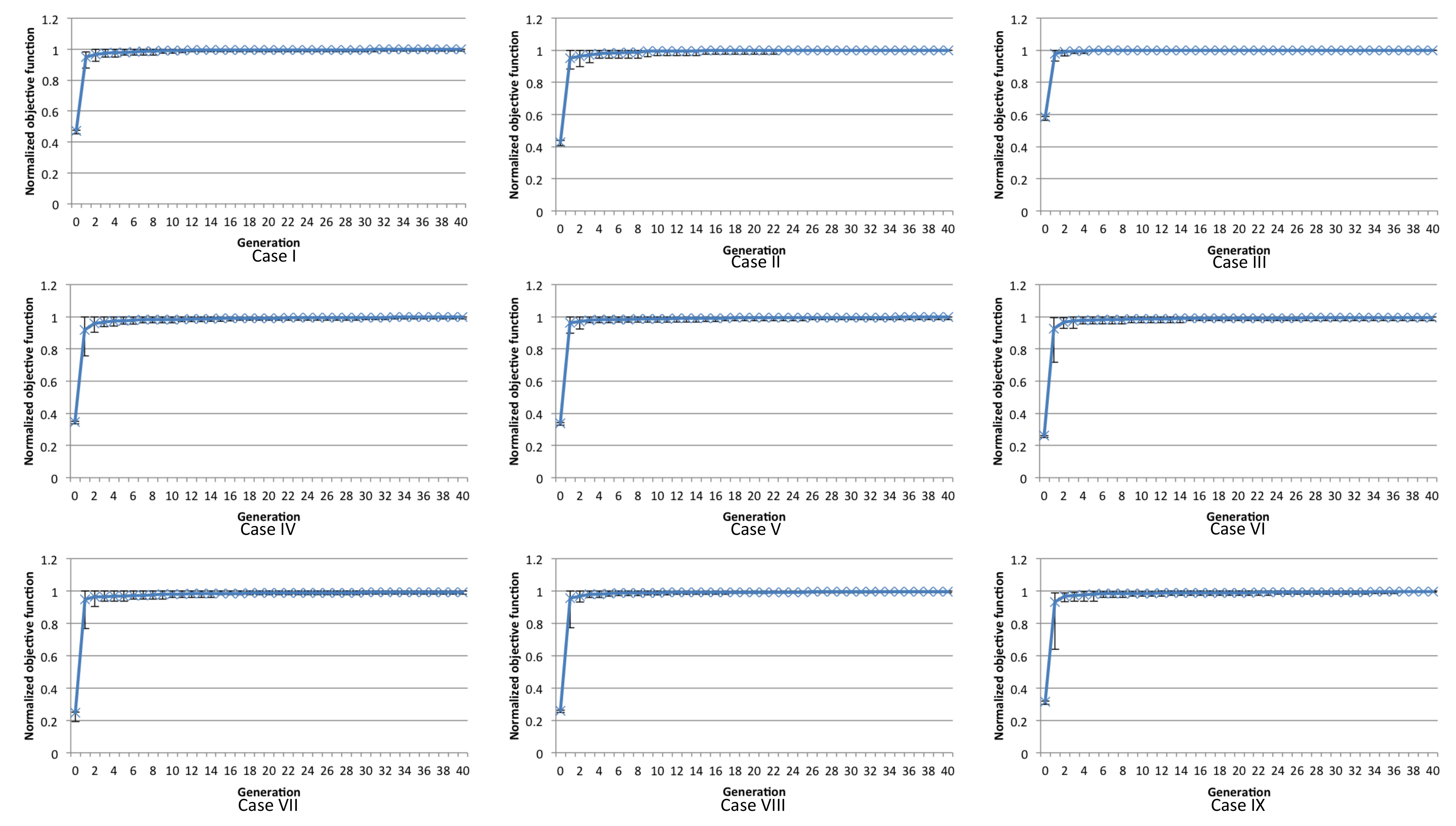} \vspace{-0.5cm}
\caption{Evolutions of the algorithm in solving admission control.} 
\label{fig:part2}
\end{figure*}

Next we investigate the performance of the algorithm \textcolor{black}{to address} admission control for an operating interval. Each fitness evaluation involves solving the scheduling problem once and the computation time for each fitness evaluation is dominated by that for scheduling. Moreover, the computation time of the algorithm depends on  the number of fitness evaluations needed. Since the population size is fixed in every generation, the run time of the algorithm can be estimated from the number of generations taken place and the results determined in Section \ref{subsec:scheduling}. Hence here we focus on the solution quality \textcolor{black}{instead}.

We run the algorithm for Cases I-IX. As we have examined all candidate solutions, we can acquire the optimal solutions of these cases. We repeat running the algorithm 20 times for each case. Fig. \ref{fig:part2} shows the average objective function value computed during the course of search for 40 generations. \textcolor{black}{As absolute values do not help reveal the performance of the algorithm,} the objective function values are \textcolor{black}{instead} normalized with the corresponding optimal values \textcolor{black}{to standardize the presentation}.\footnote{An optimal solution has the normalized objective function value equal to one.} For each data point, we also provide the error bars for the maximum and minimum values computed in the 20 repeats. The performance of the algorithm in each case is similar. The algorithm starts with relatively low quality solutions and then converges rapidly to the global optimal in a few generations. The gap between the error bars diminishes after more generations have been taken place and this further confirms the convergence of the algorithm. When the problem size increases, it takes slightly more generations to have the algorithm converged. We can conclude that our algorithm is very effective in solving the admission control problem.

\begin{figure}[!t]
\centering
\hspace{-0.3cm}
\includegraphics[width=3.5in]{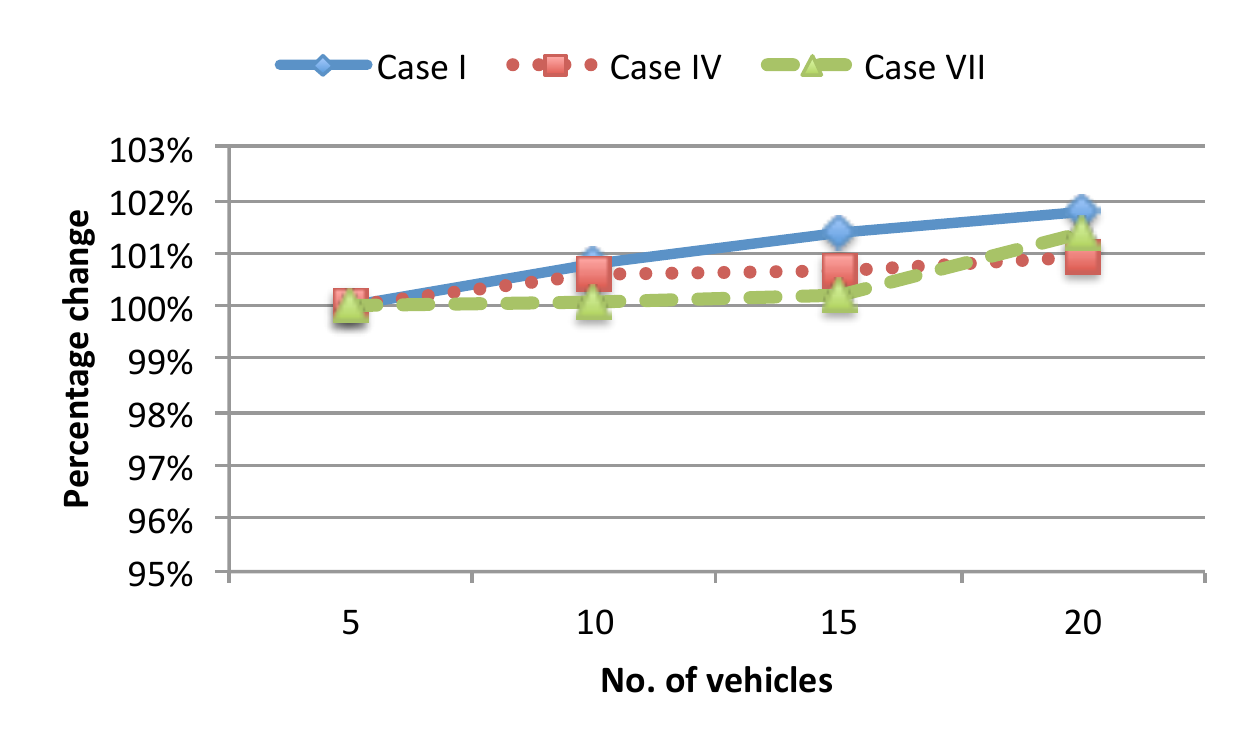} \vspace{-0.5cm}
\caption{Profits made with different numbers of vehicles.} 
\label{fig:vehtest}
\end{figure}

\textcolor{black}{We further investigate the total profits gained for the test cases with different AV population sizes. We perform the simulations with the same settings and repeat each test 20 times. Fig. \ref{fig:vehtest} shows the average results with respect to  5, 10, 15, and 20 vehicles. Since  the resultant profit highly depends on the parameters of the respective requests and vehicles, the total profits gained from different cases are not directly comparable. Instead for each case, we show the percentage change of profit by normalizing the results with the profit made with 5 AVs. Since all cases show similar trends, for clearer presentation, we give the results for Cases I, IV, and VII in Fig. \ref{fig:vehtest} only. In general, the more vehicles available, the higher profit can be made. However, the increase of profit is marginal; when compared with 5 AVs, the increase is just $1-2\%$ in the presence of 20 AVs. The reason is that more available vehicles may result in more economical routes but the total distance travelled would not be shortened significantly.}

\subsection{Admission Control in Consecutive Operating Intervals}

\begin{figure}[!t]
\centering
\hspace{-0.3cm}
\includegraphics[width=3.5in]{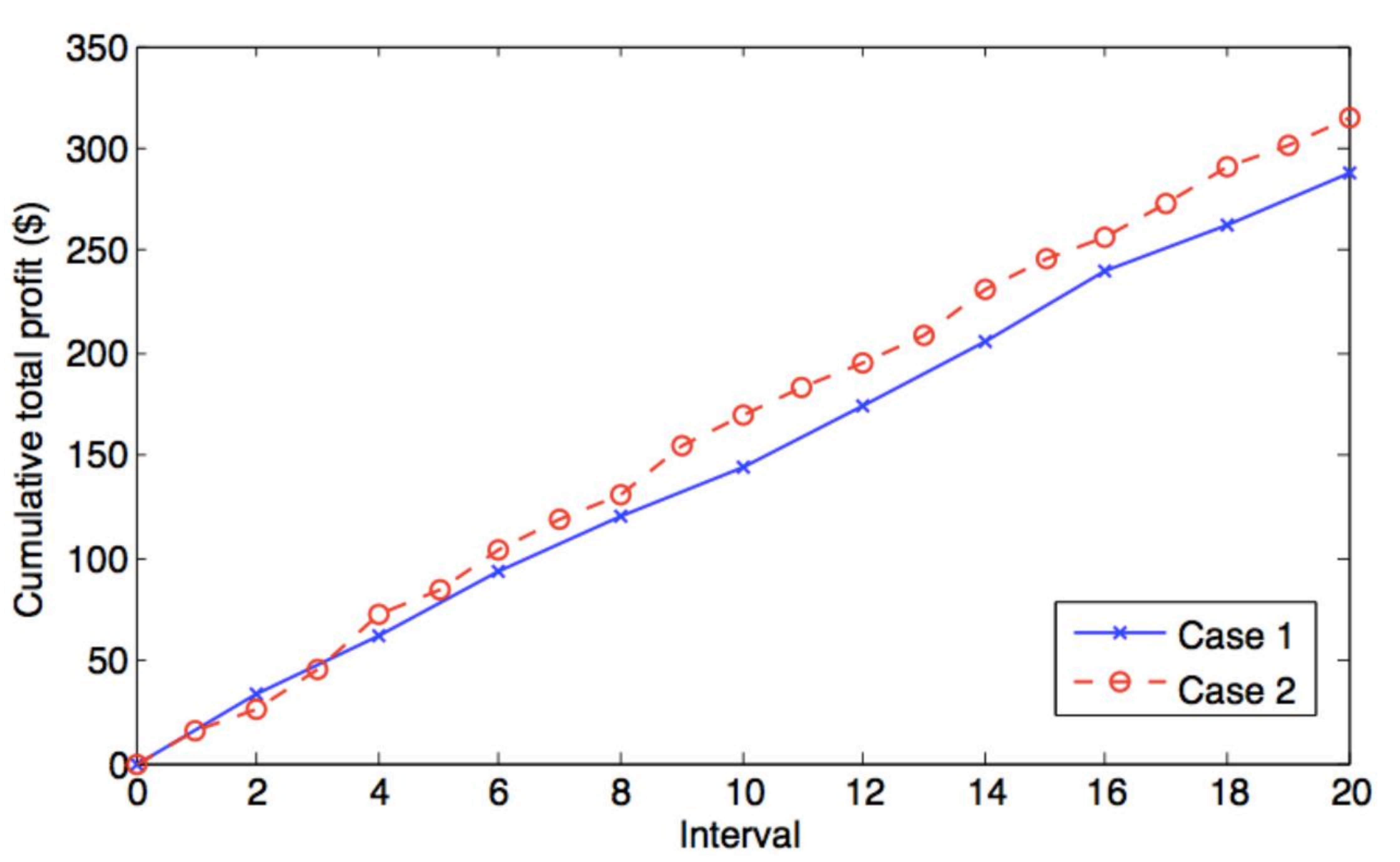} \vspace{-0.5cm}
\caption{Cumulative total profits.} 
\label{fig:profit}
\end{figure}

\begin{figure}[!t]
\centering
\hspace{-0.3cm}
\includegraphics[width=3.5in]{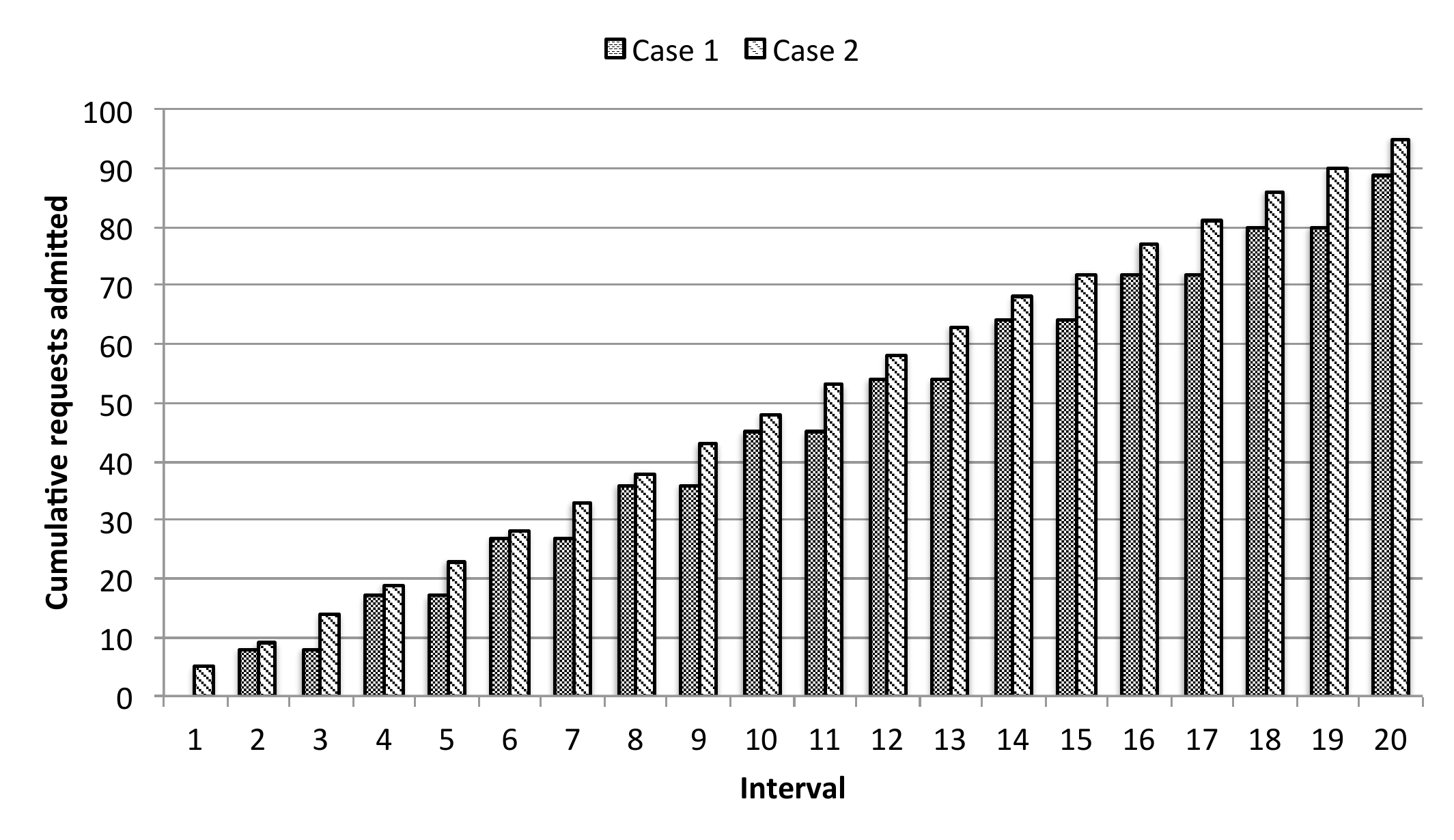} \vspace{-0.5cm}
\caption{Cumulative numbers of admitted requests.} 
\label{fig:reqadmit}
\end{figure}

Here we consider operating the system consecutively for a period of time to entertain the 100 requests in the transportation request pool. We consider two cases of different operating interval durations. In Case 1, there are 10 intervals, in each of which 10 random requests from the pool are to be scheduled. If a request is successfully admitted in an interval, it will be eliminated from the pool. Otherwise, it will be considered again in the subsequent intervals. The setting for Case 2 is similar but we consider total 20 intervals with 5 requests \textcolor{black}{being processed} in each interval. Five vehicles are arranged to serve the requests in both cases and we apply our algorithm to each interval for admission control. \textcolor{black}{In other words, we perform 10 and 20 admission control processes in Cases 1 and 2, respectively.}

Fig. \ref{fig:profit} shows the profit accumulated along the intervals\textcolor{black}{, in which we consider the duration of one interval for Case 1 is that of two intervals for Case 2}. Note that the cost is the actual expense on gas based on the traversed distance and the revenue gained from serving each request is the discounted result of having 50\% off from the real fare as if the request would be served by a normal taxi in Boston. The discount is used to compensate for the inconvenience of ride sharing and possibly longer ride time. This discount rate may be \textcolor{black}{already} attractive to many people to adopt our system instead of the normal taxi service. Hence the profit shown can be projected to a real business running in a similar scale. Fig. \ref{fig:reqadmit} provides the numbers of successfully admitted requests along the same interval horizon \textcolor{black}{as in Fig. \ref{fig:profit}}. We can see that Case 2 can produce more profit by \textcolor{black}{successfully} admitting more transportation requests. With the same number of vehicles in service, the smaller the number of requests to be scheduled in an interval, the higher the success rate of admission control is. In real situation, we normally cannot dramatically increase the size of the AV fleet and we would not intentionally reduce the number of AVs in service. On the other hand, it is much easier to adjust the number of requests to be scheduled each time by controlling the duration of each operating interval. In general, the shorter the interval, the smaller number of requests there are. Therefore, we would suggest to set the operating interval shorter, resulting in fewer requests \textcolor{black}{to be scheduled each time} and higher profits. Moreover, this will make the scheduling problem smaller by requiring shorter computation time to run the algorithm.

\section{Conclusion}\label{sec:conclusion}
With advancements in technologies, AVs become feasible and can run on the roads. \textcolor{black}{Various vehicular wireless communication technologies allow} AVs to be connected and respond cooperatively to instantaneous situations. This  constitutes a new form of public transport with high efficiency and flexibility. In this paper, we propose the AV public transportation system supporting point-to-point services with ride sharing capability. The system manages a fleet of AVs and accommodates a number of transportation requests. We focus on two major problems in the system: scheduling and admission control. The former is to configure the most economical schedules  and routes for the AVs in order to satisfy the admissible requests. 
The latter is to determine the set of admissible requests among all requests so as to produce maximum profit. We formulate the scheduling problem as \textcolor{black}{an} MILP. The admission control problem is cast as a bilevel optimization problem, in which the scheduling problem is set as a constraint. We propose a GA-based solution method to address admission control. We perform a series of simulations with a real taxi service dataset recorded in Boston and the simulation results show that our solution method is effective in solving the problem. 
\textcolor{black}{By shortening the operating intervals, the system can curtail the computation time required to solve the problem by limiting the quantity of the submitted requests and it can also produce higher profit cumulatively.}
To summarize, our contributions in this paper include: (i) designing the AV public transportation system, (ii) formulating the scheduling problem, (iii) developing distributed scheduling, (iv) formulating the admission control problem, (v) introducing the concept of admissibility and \textcolor{black}{deriving} the related analytical results, (vi) proposing an effective method to solve the admission control problem, and (vii) validating the performance of the solution method with real-world transportation service data.

\ifCLASSOPTIONcaptionsoff
  \newpage
\fi


\bibliographystyle{IEEEtran}
\bibliography{IEEEabrv}

\end{document}